\documentclass{article}
\usepackage[utf8]{inputenc}
\usepackage[total={6in, 8in}]{geometry}

\usepackage{amsmath}
\usepackage{amssymb}
\usepackage{amsthm}
\usepackage{graphicx}
\usepackage{algorithmic}
\usepackage{algorithm}
\usepackage{mathrsfs}
\usepackage{braket}
\usepackage{xcolor}
\usepackage[sort,square,numbers,compress]{natbib}
\usepackage{doi}
\usepackage{tikz}
\usetikzlibrary{quantikz}
\usepackage{hyperref}
\usepackage{cleveref}
\usepackage{pdfpages}

\newtheorem{thm}{\protect\theoremname}
\theoremstyle{plain}
\newtheorem{lem}[thm]{\protect\lemmaname}
\theoremstyle{plain}
\newtheorem{rem}[thm]{\protect\remarkname}
\theoremstyle{plain}
\newtheorem*{lem*}{\protect\lemmaname}
\theoremstyle{plain}

\theoremstyle{plain}
\newtheorem{cor}[thm]{\protect\corollaryname}

\usepackage[USenglish]{babel}
  \providecommand{\corollaryname}{Corollary}
  \providecommand{\lemmaname}{Lemma}
  \providecommand{\propositionname}{Proposition}
  \providecommand{\remarkname}{Remark}
\providecommand{\theoremname}{Theorem}

\usepackage[normalem]{ulem}

\DeclareMathOperator{\argmin}{arg min}
\DeclareMathOperator{\dist}{dist}
\newcommand{\Or}{\mathcal{O}}
\newcommand{\RR}{\mathbb{R}}
\newcommand{\E}{\mathbb{E}}

\newcommand{\wt}{\widetilde}

\newcommand{\dd}{\mathrm{d}}

\newcommand{\half}{\frac{1}{2}}
\newcommand{\ntp}[1]{\left|#1\right|_{2\pi}}
\newcommand{\one}[1]{\mathbf{1}_{\{#1\}}}

\newcommand{\He}{{H_{\mathrm{effective}}^{(1)}}}

\newcommand{\phit}{\phi}
\newcommand{\bd}{b^{\dagger}}

\newcommand{\norm}[1]{\left\|#1\right\|}
\newcommand{\abs}[1]{\left|#1\right|}

\renewcommand{\Re}{\operatorname{Re}}
\renewcommand{\Im}{\operatorname{Im}}

\usepackage{authblk}

\newcommand\CoAuthorMark{\footnotemark[\arabic{footnote}]}

\title{Heisenberg-limited Hamiltonian learning for interacting bosons}

\author[1]{Haoya Li \footnote{These authors contributed equally to this work.}}
\author[2]{Yu Tong \protect\CoAuthorMark}
\author[3]{Hongkang Ni}
\author[2]{Tuvia Gefen}
\author[1,3]{Lexing Ying}
\affil[1]{Department of Mathematics, Stanford University, Stanford, CA 94305}
\affil[2]{Institute for Quantum Information and Matter, California Institute of Technology, Pasadena, CA 91125}
\affil[3]{Institute for Computational and Mathematical Engineering, Stanford University, Stanford, CA 94305}

\begin{document}

\maketitle

\begin{abstract}
    We develop a protocol for learning a class of interacting bosonic Hamiltonians from dynamics with Heisenberg-limited scaling. For Hamiltonians with an underlying bounded-degree graph structure, we can learn all parameters with root mean squared error $\epsilon$ using $\mathcal{O}(1/\epsilon)$ total evolution time, which is independent of the system size, in a way that is robust against state-preparation and measurement error. In the protocol, we only use bosonic coherent states, beam splitters, phase shifters, and homodyne measurements, which are easy to implement on many experimental platforms. A key technique we develop is to apply random unitaries to enforce symmetry in the effective Hamiltonian, which may be of independent interest. 
\end{abstract}

\section{Introduction}


Many tasks in quantum metrology and quantum sensing can be reduced to the task of learning the Hamiltonian $H$ of a quantum system, whose evolution is described by the operator $e^{-iHt}$ \cite{de2005quantum,valencia2004distant,leibfried2004toward,bollinger1996optimal,lee2002quantum,mckenzie2002experimental,holland1993interferometric,wineland1992spin,caves1981quantum}. We call this task \emph{Hamiltonian learning}, a name that is commonly used in the literature \cite{wiebe2014a,wiebe2014b,li2020hamiltonian,che2021learning,HaahKothariTang2021optimal,yu2022,hangleiter2021,FrancaMarkovichEtAl2022efficient,ZubidaYitzhakiEtAl2021optimal,BaireyAradEtAl2019learning,GranadeFerrieWiebeCory2012robust,gu2022practical,wilde2022learnH,KrastanovZhouEtAl2019stochastic}. Besides quantum metrology and quantum sensing, Hamiltonian learning is also useful for quantum device engineering \cite{boulant2003,innocenti2020,ben2020,shulman2014,sheldon2016,sundaresan2020}, and quantum many-body physics \cite{wiebe2014a,wiebe2014b,verdon2019,burgarth2017, wang2017, kwon2020, wang2020, huang2020predicting}.

Previous works on Hamiltonian learning for many-body quantum systems are generally subject to the standard quantum limit (SQL), where to estimate the parameters in the Hamiltonian to precision $\epsilon$, $\Or(\epsilon^{-2})$ samples are required \cite{li2020hamiltonian,che2021learning,HaahKothariTang2021optimal,yu2022,hangleiter2021,FrancaMarkovichEtAl2022efficient,ZubidaYitzhakiEtAl2021optimal,BaireyAradEtAl2019learning,GranadeFerrieWiebeCory2012robust,gu2022practical,wilde2022learnH,KrastanovZhouEtAl2019stochastic,caro2022learning}. On the other hand, for simple systems such as those consisting of a single spin, the Heisenberg limit can be achieved, where to obtain $\epsilon$ precision, only $\Or(\epsilon^{-1})$ total amount of resources is needed. Achieving the Heisenberg limit requires using quantum-enhanced protocols that either use $\Or(\epsilon^{-1})$ entangled probes \cite{lee2002quantum,bollinger1996optimal,leibfried2004toward} or coherent evolution for $\Or(\epsilon^{-1})$ time \cite{de2005quantum,higgins2007entanglement,kimmel2015robust}. 
The resources consumed are the number of probes and the length of time evolution, respectively.

The natural question is, can we achieve the Heisenberg limit for many-body quantum systems? When applying the existing quantum-enhanced protocols to the many-body setting, one quickly encounters difficulties. When many entangled probes are used, one needs many copies of the quantum system with the same parameters that can evolve simultaneously without interacting with each other. It is often unclear how one can create these copies, except for certain scenarios, such as when many probes undergo evolution under the same field strength. For long coherent time-evolution, the many-body nature of the quantum systems becomes problematic as subsystems undergo open-system dynamics, and phenomena such as thermalization prevent local observables from having enough sensitivity to achieve the Heisenberg limit. One can consider performing entangled measurements across all parts of the many-body system. Still, the difficulty in simulating the system makes finding a good measurement strategy extremely difficult. 

Recently, a method was proposed in \cite{HuangTongFangYuan2023learning} to perform Hamiltonian learning for many-body spin systems with Heisenberg-limited scaling. The main technique is to apply quantum control in the form of random Pauli operators during time evolution so that the system evolves with an effective Hamiltonian that is easy to learn and, at the same time, preserves the parameters that one wants to learn. Another recent work proved that some form of quantum control is necessary for achieving the Heisenberg limit in this task \cite{dutkiewicz2023advantage}.


The above works are all focused on multi-qubit systems, and Heisenberg-limited Hamiltonian learning for bosonic systems is relatively less studied. 
Bosonic systems, such as superconducting circuits \cite{krantz2019quantum,blais2021circuit}, integrated photonic circuits \cite{wang2020integrated} and optomechanical platforms \cite{aspelmeyer2014cavity, metcalfe2014applications} are widely used for quantum sensing, communication, and computing \cite{blais2004cavity,clerk2010introduction,chamberland2022building,adhikari2014gravitational}. These quantum applications require efficient calibration \cite{hangleiter2021}, and it is thus highly desirable to develop optimal algorithms for characterizing bosonic Hamiltonians. For example, quantum computing and sensing with transmons require learning the energy levels and interactions between the transmons and microwave resonators. 

For bosonic systems, there is a different set of ``easy'' quantum states, unitaries, and measurements than for spins. This work assumes that one can prepare coherent states, apply phase shifters and beam splitters, and perform the homodyne measurement. We note that although we may use terms from quantum optics, such as ``phase shifters'', we do not constrain our discussion to the optical setting. Additionally, in our protocol, we do not require any squeezing, which can be experimentally difficult to implement \cite{qin2022beating,dassonneville2021dissipative,kronwald2013arbitrarily,wiseman1994squeezing}. Using these resources, we present a protocol to learn a class of interacting bosonic Hamiltonians with Heisenberg-limited scaling. These Hamiltonians involve terms that are quadratic or quartic in the creation and annihilation operators, and are particle-number preserving. The specific form of the Hamiltonians is given in \eqref{eq:hamiltonian_general}. Our protocol can also tolerate a constant amount of noise in the state preparation and measurement (SPAM) procedures and has a small classical post-processing cost.

In our method, we apply random unitaries during time evolution to reshape the Hamiltonian into an effective Hamiltonian that is easier to learn. This follows the same high-level idea as \cite{HuangTongFangYuan2023learning} but is specifically tailored to the bosonic setting. Moreover, we can interpret the procedure as enforcing a target symmetry in the effective Hamiltonian, thus putting constraints on the dynamics. We believe this technique may be useful for other problems in quantum simulation as well \cite{nguyen2022digital,halimeh2021gauge}. In analyzing the deviation from the effective dynamics, the unboundedness of the bosonic Hamiltonian terms poses a challenge, as the analysis in \cite{HuangTongFangYuan2023learning} requires Hamiltonian terms to be bounded. We use more involved techniques to overcome this difficulty in Section~\ref{sec:deviation_from_effective_dynamics_appendix}.


\section{Results}


In this work, we focus on quantum systems on $N$ bosonic modes forming a $d$-dimensional lattice, with the Hamiltonian of the form
\begin{equation}
    \label{eq:hamiltonian_general}
    H = \sum_{\braket{i,j}} h_{ij}b_i^{\dagger}b_j + \sum_i \omega_i b_i^{\dagger}b_i + \frac{\xi_i}{2}\sum_i n_i(n_i-1),
\end{equation}
where $b_i$ ($b_i^{\dagger}$) are bosonic annihilation (creation) operators, and $n_i=b_i^{\dagger}b_i$ are the number opreators. $\braket{i,j}$ means that the summation is over sites $i,j$ that are adjacent to each other. $h_{ij}=h_{ji}^*$, and each $\xi_i$ and $\omega_i$ is a real number. We also assume that $|h_{ij}|,|\omega_i|,|\xi_i|\leq 1$. This class of Hamiltonians is relevant for superconducting quantum processors \cite{hangleiter2021}, arrays of coupled cavities \cite{hartmann2008quantum}, and phonon dynamics in ion crystals \cite{porras2004bose,shen2014scalable}. We will present a protocol that generates estimates $\hat{h}_{ij}$, $\hat{\omega}_i$, and $\hat{\xi}_i$ such that
\begin{equation}
    \mathbb{E}[|\hat{h}_{ij}-h_{ij}|^2],\ \mathbb{E}[|\hat{\omega}_i-\omega_i|^2],\ \mathbb{E}[|\hat{\xi}_i-\xi_i|^2]\leq \epsilon^2,
\end{equation}
for all $i$ and $j$.
The protocol has the following properties:
\begin{enumerate}
    \item The \emph{total evolution time} is $\Or(\epsilon^{-1})$;
    \item The number of experiments is $\Or(\operatorname{polylog}(\epsilon^{-1}))$;
    \item A constant amount of SPAM error can be tolerated.
\end{enumerate}
More precisely, our protocol consists of $N_{\exp}=\Or(\operatorname{polylog}(\epsilon^{-1}))$ experiments, which we number by $1,2,\cdots,N_{\exp}$.
In the $j$th experiment, we will initialize each bosonic mode in the system in a coherent state, let the system evolve for time $t_j>0$, and perform homodyne measurement for the bosonic modes. During time evolution, we will apply random beam splitters (on two modes) or phase shifters (on one mode). The total evolution time is defined to be $\sum_{j=1}^{N_{\exp}}t_j$, which is the amount of time required to run all the experiments. We assume that after we prepare the initial state and before we perform the measurement, the system goes through error channels $\mathcal{E}_1$ and $\mathcal{E}_2$, which model the SPAM error. If $\|\mathcal{E}_1-\mathcal{I}\|_{\diamond}+\|\mathcal{E}_2-\mathcal{I}\|_{\diamond}$ is upper-bounded by a small constant, then our protocol will still be able to reach arbitrary precision $\epsilon$. Here $\|\cdot\|_{\diamond}$ is the diamond norm \cite{watrous2018theory}, and $\mathcal{I}$ is the identity channel. The precision is measured by the mean squared error (MSE). We are using the big-$\Or$ notation to hide the constants for simplicity, and we note that these constants never depend on the system size. Our protocol generates $\Or(NN_{\exp})=\Or(N\operatorname{polylog}(\epsilon^{-1}))$ classical data and it takes a similar amount of time to process these data to compute the estimates.

Below we will describe the protocol in detail. We will start with a protocol to learn a single anharmonic oscillator, which forms the basic building block for more complex situations.

\subsection{Learning an anharmonic oscillator}
\label{sec:learning_an_anharmonic_oscillator}

We first consider the simple case in which
\begin{equation}
    H_{\mathrm{AHO}} = \omega n + \frac{\xi}{2}n(n-1),
\end{equation}
where $n=b^{\dagger}b$. We want to estimate the coefficients $\omega$ and $\xi$ with root mean squared error (RMSE) at most $\epsilon$.


This is a quantum sensing problem with two parameters to be estimated. In quantum sensing, one usually calculates the quantum Cram{\'e}r-Rao bound (QCRB) that provides a lower bound on the MSE of unbiased estimators. Because the two parameters correspond to Hamiltonian terms that commute with each other, the QCRB scales inverse quadratically with time, allowing us to achieve the Heisenberg-limited scaling. This bound, however, is valid only for local estimation where the prior distribution of the estimators is already concentrated around the exact value. Here we provide an estimation protocol that achieves this scaling without any prior knowledge of the parameters.

Our protocol builds upon a robust frequency estimation algorithm similar to the robust phase estimation algorithm proposed in \cite{kimmel2015robust} as well as the alternative version in \cite{ni2023low}. In the robust phase estimation algorithm, we assume that through performing certain experiments that we will specify when introducing our protocol, we have access to a random variable $Z_{\delta}(t)$ from measurement results, such that $|Z_{\delta}(t)-e^{-i\omega t}|\leq 1$ with probability at least $1-\delta$, and generating such a random variable requires evolution time $\Or(t\log(\delta^{-1}))$. With multiple samples of this variable for different values of $t$ and $\delta$, we can generate an estimate of $\omega$ with RMSE at most $\epsilon$ using $\Or(\epsilon^{-1})$ total evolution time. The algorithm proceeds by iteratively obtaining estimates with increasing accuracy through longer time evolution until the target precision is achieved. A detailed description of the algorithm and proof of its correctness can be found in Section~\ref{sec:robust_frequency_estimation_appendix}. 


We initialize the system in a coherent state $\ket{\alpha}=e^{-|\alpha|^2/2}\sum_{k}(\alpha^k/\sqrt{k!})\ket{k}$, and let the system evolve under the Hamiltonian $H_{\mathrm{AHO}}$. In the end we perform homodyne measurements with quadrature operators $X=(b+b^{\dagger})/\sqrt{2}$ and $P=i(b^{\dagger}-b)/\sqrt{2}$ in separate experiments. With these measurement results we will be able to estimate $\braket{b}_{\alpha,t}=\braket{\alpha|e^{iH_{\mathrm{AHO}}t}be^{-iH_{\mathrm{AHO}}t}|\alpha}$, which can be exactly computed to be
\begin{equation}
\label{eq:anharmonic_oscillator_expectation}
    \braket{b}_{\alpha,t} = \alpha e^{-|\alpha|^2} e^{-i\omega t} e^{|\alpha|^2 e^{-i\xi t}}.
\end{equation}
We perform this calculation in Section~\ref{sec:learning_an_anharmonic_oscillator_appendix}.

Using \eqref{eq:anharmonic_oscillator_expectation}, we can extract the values of $\omega$ and $\xi$ from $\braket{b}_{\alpha,t}$. For $\omega$, note that $\braket{b}_{\alpha,t}/\alpha = e^{-i\omega t} + \Or(|\alpha|^2)$, and therefore we can choose $|\alpha|$ to be below a small constant so that an estimate for $\braket{b}_{\alpha,t}/\alpha$ will be close to $e^{-i\omega t}$ within some small constant distance, which enables us to apply the robust frequency estimation algorithm to estimate $\omega$ with RMSE at most $\epsilon$ using total evolution time $\Or(\epsilon^{-1})$.

For $\xi$, we can extract its value by constructing a periodically oscillating signal through
\begin{equation}
    e^{-i\xi t} = \frac{1}{|\alpha_1|^2-|\alpha_2|^2}\log\left(\frac{\alpha_2\braket{b}_{\alpha_1,t}}{\alpha_1\braket{b}_{\alpha_2,t}}\right) + 1.
\end{equation}
This enables us to estimate $\xi$ using the robust frequency estimation algorithm. Note that, once again, $\braket{b}_{\alpha_1,t}$ and $\braket{b}_{\alpha_2,t}$ only need to be estimated to constant precision, rather than $\epsilon$ precision which would result in an $\Or(\epsilon^{-2})$ scaling that would destroy the Heisenberg-limited scaling.

In the above procedure, we need to estimate the expectation of $X$ and $P$ operators, which are unbounded operators that can infinitely amplify any error in the quantum state. Fortunately, we found that we can replace them with operators $X\one{|X|\leq M}$ and $P\one{|P|\leq M}$, where $\one{|X|\leq M}=\int_{|x|\leq M} \ket{x}\bra{x}\dd x$ and $\one{|P|\leq M}$ is similarly defined. This means truncating the eigenvalues of these operators at a threshold $M=\Or(1)$. In practice, we can simply discard any $X$ and $P$ samples that are above the threshold $M$ to implement the measurement associated with these truncated operators. This fact, together with the error tolerance in the robust frequency estimation algorithm, enables us to tolerate a constant amount of error from SPAM and time evolution. 
The combined error from all sources should be below a small constant, which is sufficient for achieving arbitrarily high precision. 



\subsection{Learning two coupled anharmonic oscillators}
\label{sec:learning_two_coupled_anharmonic_oscillators}

Next, we consider a system consisting of two coupled anharmonic oscillators, where the Hamiltonian is of the following form:
\begin{equation}
\label{eq:two_mode_hamiltonian}
    H = \omega_1 b_1^{\dagger}b_1 + \omega_2 b_2^{\dagger}b_2 + h_{12}b_1^{\dagger}b_2 + h_{21}b_2^{\dagger}b_1 + \frac{\xi_1}{2}n_1(n_1-1) + \frac{\xi_2}{2}n_2(n_2-1)
\end{equation}
The goal is to learn all the coefficients $\omega_{1}$, $\omega_{2}$, $\xi_{1}$, $\xi_{2}$, and $h_{12}$ ($h_{21}=h^*_{12}$).

We first focus on learning the single-mode coefficients $\omega_1$, $\omega_2$, $\xi_1$, and $\xi_2$. To do this, we will insert random unitaries during time evolution to decouple the bosonic modes from each other. In other words, the time evolution operator undergoes the following transformation
\begin{equation}
    \label{eq:time_evolution_transformation}
    e^{-iHt} \mapsto \prod_{j=1}^{r} U_j^{\dagger}e^{-iH\tau} U_j = \prod_{j=1}^{r} e^{-iU_j^{\dagger}HU_j\tau},
\end{equation}
where the $U_j$, $j=1,2,\cdots,r$, are the random beam splitters or phase shifters that we insert, $r=t/\tau$, and the product goes from right to left. Each $U_j$ is independently drawn from a distribution that we denote by $\mathcal{D}$. In the limit of $\tau\to 0$, the dynamics can be described by an effective Hamiltonian
\begin{equation}
    \label{eq:effective_hamiltonian}
    H_{\mathrm{effective}} = \mathbb{E}_{U\sim \mathcal{D}} U^{\dagger}HU.
\end{equation}
This can be seen by considering the Taylor expansion of the time-evolved state in a small time step:
\begin{equation}
\label{eq:first_order_approx_short_time_step}
\begin{aligned}
    \mathbb{E}_{U\sim\mathcal{D}}[e^{-iU^{\dagger}HU\tau}\rho e^{iU^{\dagger}HU\tau}] &= \rho - i\tau \mathbb{E}_{U\sim\mathcal{D}}[[U^{\dagger}HU,\rho]] + \Or(\tau^2) \\ 
    &= e^{-i\mathbb{E}_{U\sim\mathcal{D}}[U^{\dagger}HU]\tau}\rho e^{i\mathbb{E}_{U\sim\mathcal{D}}[U^{\dagger}HU]\tau} + \Or(\tau^2).
\end{aligned}
\end{equation}
Note that the above is not a rigorous proof, because the $\Or(\tau^2)$ residue is an unbounded operator. We provide a rigorous bound of how far the actual dynamics deviate from the limiting effective dynamics with finite $\tau>0$ in Section~\ref{sec:deviation_from_effective_dynamics_appendix}. 
The above procedure introduces additional randomness to our protocol, but it does not introduce any sample complexity overhead, because we only need the final quantum states to be close in terms of the trace distance.

To learn all the single mode coefficients, we let the unitary $U$ drawn from the distribution $\mathcal{D}$ be 
\begin{equation}
    U = e^{-i\theta n_1},\quad \theta\sim\mathcal{U}([0,2\pi]).
\end{equation}
Here $\mathcal{U}([0,2\pi])$ is the uniform distribution over $[0,2\pi]$.
We can then compute the effective Hamiltonian
\begin{equation}
\label{eq:effective_hamiltonian_single_mode}
    H_{\mathrm{effective}} = \frac{1}{2\pi}\int_0^{2\pi} e^{i\theta n_1}He^{-i\theta n_1}\dd \theta = \omega_1 n_1 + \omega_2 n_2 + \frac{\xi_1}{2}n_1(n_1-1) + \frac{\xi_2}{2}n_2(n_2-1).
\end{equation}
In other words, the coupling term $h_{12}b_1^{\dagger}b_2 + h_{21}b_2^{\dagger}b_1$ is cancelled in the process, due to the equality $e^{i\theta n_1}b_1 e^{-i\theta n_1}=e^{-i\theta}b_1$.
We can interpret this procedure as enforcing a particle number conservation on the first bosonic mode. In the effective Hamiltonian, the two bosonic modes are no longer coupled together, and therefore we can apply the learning algorithm described in Section~\ref{sec:learning_an_anharmonic_oscillator} to learn the parameters of the two modes separately. For a more detailed description of the protocol see Section~\ref{sec:single_mode_coeff_appendix}.

Next, we will learn the coupling coefficient $h_{12}$. We will use the following unitaries
\begin{equation}
    U_{x}(\theta) = e^{i\theta (b_1^{\dagger}b_2+b_2^{\dagger}b_1)},\quad U_{y}(\theta) = e^{\theta (b_1^{\dagger}b_2-b_2^{\dagger}b_1)}.
\end{equation}
Our protocol is based on the observation that under a single-particle basis rotation, $h_{12}$ can be estimated from the new single-mode coefficients. More precisely, we let $\tilde{b}_1 = U_y(\pi/4)b_1 U_y^{\dagger}(\pi/4)$, $\tilde{b}_2 = U_y(\pi/4)b_2 U_y^{\dagger}(\pi/4)$, and the new bosonic modes will be related to the old ones through
\begin{equation}
\label{eq:single_particle_basis_rotation_Y}
    \begin{pmatrix}
        \tilde{b}_1 \\
        \tilde{b}_2
    \end{pmatrix}
    =
    \begin{pmatrix}
        \cos(\pi/4) & \sin(\pi/4) \\
        -\sin(\pi/4) & \cos(\pi/4)
    \end{pmatrix}
    \begin{pmatrix}
        b_1 \\
        b_2
    \end{pmatrix}.
\end{equation}
We will then rewrite the Hamiltonian \eqref{eq:two_mode_hamiltonian} in terms of $\tilde{b}_1$ and $\tilde{b}_2$. The quadratic part of $H$ can be written as
$\tilde{\omega}_1 \tilde{b}_1^{\dagger}\tilde{b}_1 + \tilde{\omega}_2 \tilde{b}_2^{\dagger}\tilde{b}_2 + \tilde{h}_{12}\tilde{b}_1^{\dagger}\tilde{b}_2 + \tilde{h}_{21}\tilde{b}_2^{\dagger}\tilde{b}_1$, where
\begin{equation}
\label{eq:omega_tilde_expression}
    \tilde{\omega}_1 = \frac{\omega_1+\omega_2}{2}+\Re h_{12}.
\end{equation}
Therefore, $\Re h_{12}$ can be estimated if we can learn $\tilde{\omega}_1$.
The quartic part becomes more complicated, but the procedure we describe next will yield an effective Hamiltonian of a simpler form.

In our protocol for learning $\Re h_{12}$, we will let the random unitaries $U_j$ in \eqref{eq:time_evolution_transformation} be 
\begin{equation}
    U_j = U_x(-\theta/2),\quad \theta\sim\mathcal{U}([0,2\pi]),
\end{equation}
where $\mathcal{U}([0,2\pi])$ denotes the uniform distribution on $[0,2\pi]$. Note that $e^{-i\theta \tilde{n}_1}=e^{-i(\theta/2)(n_1+n_2)}U_x(-\theta/2)$ where $\tilde{n}_1=\tilde{b}_1^{\dagger}\tilde{b}_1$, and because the total particle number $n_1+n_2$ is conserved, the random unitary $U_x(-\theta/2)$ is equivalent to $e^{-i\theta \tilde{n}_1}$ up to a global phase. This random unitary, as in \eqref{eq:effective_hamiltonian_single_mode}, results in an effective Hamiltonian in which $\tilde{n}_1$ is conserved. The effective Hamiltonian can be written as the following
\begin{equation}
    H_{\mathrm{effective}} = \tilde{\omega}_1 \tilde{n}_1 + \tilde{\omega}_2 \tilde{n}_2 + \frac{\tilde{\xi}_{11}}{2}\tilde{n}_1(\tilde{n}_{1}-1) + \frac{\tilde{\xi}_{22}}{2}\tilde{n}_2(\tilde{n}_2-1) + \tilde{\xi}_{12}\tilde{n}_1\tilde{n}_2.
\end{equation}
In this effective Hamiltonian, the two bosonic modes $\tilde{b}_1$ and $\tilde{b}_2$ are still coupled through the term $\tilde{\xi}_{12}\tilde{n}_1\tilde{n}_2$. However, because the particle numbers on both modes are conserved, we can simply initialize the system with no particle on the mode $\tilde{b}_2$, and the coupling term will have no effect. More specifically, the initial state we use is $U_y(\pi/4)\ket{\alpha}\ket{0}$, which is an $\alpha$-eigenstate for $\tilde{b}_1$ and a $0$-eigenstate for $\tilde{b}_2$. The effective Hamiltonian can then be further reduced to
\begin{equation}
    H_{\mathrm{effective}}' = \tilde{\omega}_1 \tilde{n}_1 + \frac{\tilde{\xi}_{11}}{2}\tilde{n}_1(\tilde{n}_{1}-1).
\end{equation}
This enables us to learn $\tilde{\omega}_1$ using the single-mode protocol in Section~\ref{sec:learning_an_anharmonic_oscillator}, which then gives us $\Re h_{12}$ through \eqref{eq:omega_tilde_expression}. When performing homodyne measurement in the end, we also need to apply $U_y(-\pi/4)$ to rotate back to the original single-particle basis. We write down the quantum state we get right before measurement to summarize the whole procedure:
\begin{equation}
    U_y\left(-\frac{\pi}{4}\right)\prod_{j=1}^r\left(U_x\left(\frac{\theta_j}{2}\right)e^{-iH\tau}U_x\left(-\frac{\theta_j}{2}\right)\right)U_y\left(\frac{\pi}{4}\right)\ket{\alpha}\ket{0},
\end{equation}
where all $\theta_j$ are independently drawn from the uniform distribution over $[0,2\pi]$.

The above procedure yields $\Re h_{12}$. For $\Im h_{12}$, we only need to switch the roles of $U_x(\theta)$ and $U_y(\theta)$ and go through the same procedure. For a more detailed discussion, see Section~\ref{sec:coupling_coeff_appendix}. 


\subsection{Learning an $N$-mode system}
\label{sec:learning_an_N_mode_system}

So far, we have concerned ourselves with learning small systems with one or two modes, but the protocol we develop can be easily generalized to $N$-mode systems. This section will focus on $N$ bosonic modes arranged on a 1D chain. For the more general situation with a bounded degree graph, e.g., $D$-dimensional square lattice, Kagome lattice, etc., see Section~\ref{sec:divide_and_conquer_for_N_mode_appendix}. 
The Hamiltonian is described by \eqref{eq:hamiltonian_general}, where the bosonic modes are labeled $1,2,\cdots, N$, and $i$ and $j$ are adjacent only when $j=i\pm 1$.

For this $N$-mode system, we consider a divide-and-conquer approach. We will apply random unitaries so that in the effective dynamics, the system is divided into clusters of one or two modes, each of which does not interact with the rest of the system. In this way, we can learn the parameters associated with each cluster independently and in parallel using our protocol in Section~\ref{sec:learning_two_coupled_anharmonic_oscillators}. 

More specifically, we apply random unitaries in the same way as described in \eqref{eq:time_evolution_transformation}. The random unitary $U_j$ is first chosen to be
\begin{equation}
    U_j = \prod_{k=1}^{\lfloor N/3\rfloor} e^{-i\theta_{3k} n_{3k}},
\end{equation}
where the random variables $\theta_{3k}$ are independently drawn from $\mathcal{U}([0,2\pi])$, the uniform distribution over $[0,2\pi]$. Randomly applying the unitaries from this distribution enforces particle number conservation on sites with indices that are integer multiples of $3$. Therefore, any Hamiltonian term $b_i^{\dagger}b_j$ that involves sites $3, 6, 9,\cdots$ are canceled. The effective Hamiltonian is
\begin{equation}
    \begin{aligned}
        H &= \omega_1 n_1 + \omega_2 n_2 + h_{12}b_1^{\dagger}b_2 + h_{21}b_2^{\dagger}b_1 \\
        &+ \omega_4 n_4 + \omega_5 n_5 + h_{45}b_4^{\dagger}b_5 + h_{54}b_5^{\dagger}b_4 \\
        &+ \cdots \\
        &+\sum_{i} \frac{\xi_i}{2}n_i(n_i-1),
    \end{aligned}
\end{equation}
where we did not include the terms $\omega_3 n_3$, $\omega_6 n_6$, etc., because they only contribute a global phase.
In this Hamiltonian, the two modes $1$ and $2$ form a cluster: they only interact with each other but not with the rest of the system. The same is true for modes 5 and 6, 7 and 8, etc. We can then apply the two-mode protocol in Section~\ref{sec:learning_two_coupled_anharmonic_oscillators} to learn all coefficients associated with modes 1, 2, 5, 6, ... Note that coefficients associated with different clusters can be learned in parallel in the same experiment.

Other coefficients remain to learn, such as $\omega_3$, $h_{23}$, and $h_{34}$. We can adopt the same strategy but choose the random unitary $U_j = \prod_{k=0}^{\lfloor N/3\rfloor-1} e^{-i\theta_{3k+1} n_{3k+1}}$ so that modes 2 and 3, 5 and 6, etc. now form clusters. Similarly, we can let modes 3 and 4, 6 and 7, etc., form clusters. In this way, we can learn all the coefficients in the Hamiltonian using three different clustering schemes. The total evolution time required for carrying out all experiments will only be three times the cost of a two-mode protocol because different clusters can be learned in parallel.

More generally, we consider a system whose interaction can be described by a bounded-degree graph. We can design similar clustering schemes based on an appropriate coloring of its link graph, i.e., the graph whose vertices are the edges of the original graph. The overhead introduced will be quadratic in the degree of the original graph and independent of the system size $N$. This is discussed in more detail in Section~\ref{sec:divide_and_conquer_for_N_mode_appendix}. 

\section{Discussion}

In this work, we propose a protocol to learn a class of interacting bosonic Hamiltonians with Heisenberg-limited scaling. Our protocol uses only elements of linear optics that that can be implemented on various experimental platforms. Besides achieving the Heisenberg-limited scaling, our protocol can also tolerate a constant amount of 
SPAM noise thanks to the robust frequency estimation subroutine discussed in Section~\ref{sec:learning_an_anharmonic_oscillator}. As a part of the protocol, we also propose a method to enforce symmetry on the effective Hamiltonian governing the system's evolution as discussed in more detail in Section~\ref{sec:enforcing_symmetry_using_random_unitaries}.

To our knowledge, our work is the first to propose a method that learns interacting bosonic Hamiltonians with Heisenberg-limited scaling in a scalable way. However, many open problems remain to be solved in this research direction. In this work, we only consider the particle-number preserving Hamiltonian in \eqref{eq:hamiltonian_general}, but realistic Hamiltonians may contain terms that do not preserve the particle number, such as the coupling term in the Jaynes–Cummings model \cite{jaynes1963comparison} and capacitive and inductive couplings between superconducting circuits \cite{krantz2019quantum}. Also, higher-order anharmonic effects beyond the fourth order may be non-negligible in certain quantum systems.  

In our protocol, we need to apply random unitaries with a frequency that depends on the target precision. For higher precision, the speed of applying these unitaries will also need to be faster, which may be a problem for experimental implementation. A possible solution is to use some form of continuous control as considered in \cite{stark2017narrow,gordon2008optimal,fanchini2007continuously,bermudez2012robust,xu2012coherence,dutkiewicz2023advantage}. Moreover, since our protocol requires letting the system evolve coherently for $\Or(\epsilon^{-1})$ times to reach $\epsilon$ precision, the achievable precision will be limited by quantum noise such as dephasing and photon losses that limit the coherence time of most experimental Bosonic systems. 
It would be therefore interesting to explore whether noise suppression techniques such as dynamical decoupling \cite{sekatski2016dynamical, schmitt2017submillihertz,arenz2017dynamical,heinze2019universal} and quantum error correction \cite{dur2014improved, arrad2014increasing,kessler2014quantum,herrera2015quantum,demkowicz2017adaptive,Zhou2017AchievingTH} can mitigate this limitation and whether they can be incorporated into our protocol in a useful and scalable way.

Random Clifford unitaries played a crucial role in the classical shadow formalism \cite{huang2020predicting,elben2022randomized} as well as Hamiltonian learning \cite{HuangTongFangYuan2023learning}. Similarly, one may wonder whether the random gaussian unitaries used in this work can be useful for other quantum information tasks for bosonic systems, such as classical shadow tomography for continuous-variable systems \cite{gu2022efficient,becker2022classical,gandhari2022continuous}.

\section{Methods}

\subsection{Enforcing symmetry using random unitaries}
\label{sec:enforcing_symmetry_using_random_unitaries}


This section will describe how to enforce symmetry using random unitaries. This strategy is similar in spirit to the symmetry protection strategies in \cite{nguyen2022digital,halimeh2021gauge}, but is easier to scale to an $N$-mode system in the current setting.

Let us first consider the general case where we have a compact Lie group $G$ that describes the symmetry we want in the quantum system. Our quantum system is evolving under a Hamiltonian $H$ that does not necessarily satisfy this symmetry, i.e., there may exist $g\in G$ such that $gHg^{-1}\neq H$ (here we equate an element of the Lie group with its matrix representation). We want to have the system evolve under an effective Hamiltonian $H_{\mathrm{effective}}$ that satisfies the symmetry, i.e.,
\begin{equation}
    gH_{\mathrm{effective}}g^{-1} = H_{\mathrm{effective}}.
\end{equation}

We achieve this by inserting random unitaries in the same way as in \eqref{eq:time_evolution_transformation}, which gives us an effective Hamiltonian according to \eqref{eq:effective_hamiltonian}. The distribution from which we draw the random unitaries is the Haar measure on $G$, which we denote by $\mu$. The  effective Hamiltonian can be computed as
\begin{equation}
    H_{\mathrm{effective}} = \int gHg^{-1}\mu(\dd g).
\end{equation}
When the Hamiltonian $H$ is unbounded, the above equality may only hold in a weak sense.
We can verify that this effective Hamiltonian satisfies the desired symmetry because
\begin{equation}
    g' H_{\mathrm{effective}} g'^{-1} = \int g'gH(g'g)^{-1}\mu(\dd g) = \int g'gH(g'g)^{-1}\mu(\dd (g' g)) = H_{\mathrm{effective}}.
\end{equation}
Here we have used the property of the Haar measure that $\mu(\dd (g' g))=\mu(\dd g)$.

It may not be easy to randomly apply elements from the symmetry group $G$. Still, in our learning protocol, we will only enforce symmetries that are either U(1) or $\mathrm{U}(1)\times \mathrm{U}(1)\times\cdots\times \mathrm{U}(1)=U(1)^{\times N}$, where sampling can easily be done for each U(1) group separately.

\section{Acknowledgements}
The authors thank Matthias Caro for helpful discussions.
Y.T. acknowledges funding from the U.S. Department of Energy Office of Science, Office of Advanced Scientific Computing Research (DE-NA0003525, and DE-SC0020290). Work supported by DE-SC0020290 is supported by the DOE QuantISED program through the theory consortium ``Intersections of QIS and Theoretical Particle Physics'' at Fermilab. The Institute for Quantum Information and Matter is an NSF Physics Frontiers Center. The work of H.L. and L.Y. is partially supported by National Science Foundation under awards DMS-2011699 and DMS-2208163. T.G. acknowledges funding provided by the Institute for Quantum Information and Matter and the Quantum Science and Technology Scholarship of the Israel Council for Higher Education.

\bibliographystyle{unsrtnat}
\bibliography{ref}

\appendix
\section{Robust frequency estimation}
\label{sec:robust_frequency_estimation_appendix}

Our main tool to achieve the Heisenberg limit is an algorithm to estimate the frequency from a complex-valued signal with Heisenberg-limited scaling. This algorithm resembles the robust phase estimation algorithm in \cite{kimmel2015robust} but is different in that we can deal with cases where the frequency we want is encoded in the expectation value rather than the probability.

More precisely, we assume access to a signal $Z_{\delta}(t)$ that is close to $e^{-i\omega t}$, where $|\omega|<W$, by a constant amount of error in both the phase and the amplitude with probability $1-\delta$, where $\delta$ can be tuned. It is also reasonable to assume that for smaller $\delta$, generating the corresponding $Z_{\delta}(t)$ will be more costly, i.e., requiring longer evolution time. Our algorithm then uses $Z_{\delta}(t)$ for different values of $\delta$ and $t$ to refine the estimation of $\omega$ iteratively. In each iteration, we use the result from the previous iteration to get an estimate $\theta_j$ satisfying
\begin{equation}
    \omega/\wt{W}\in \left(\theta_j-\frac{\pi}{3\cdot2^j}, \theta_j+\frac{\pi}{3\cdot2^j}\right),~ \mod 2\pi,
\end{equation}
where 
$\wt{W}=3W/\pi$
is a normalization factor. A detailed algorithm description can be found in Algorithm~\ref{alg:pruning}, adapted from \cite{ni2023low}. 


\begin{algorithm}[ht]
	\caption{Robust frequency estimation}
	\label{alg:pruning}
	\begin{algorithmic}
		\STATE{\textbf{Input:} $\epsilon$: target accuracy, $\delta$: upper bound of the failure probability, $W$: upper bound for $\omega$. }
		\STATE{Let $\wt{W} = \frac{3W}{\pi}$ and $J = \max\{1, \lceil \log_2(4\pi \wt{W}/(3\epsilon))\rceil$\}.}
		\STATE{Set $\theta_{-1} = 0$.}
		\FOR{$j = 0,1,\ldots,J-1$}
        \STATE{Set $t_j=2^j/\wt{W}$.}
        \STATE{Calculate $\delta_j= \frac{27\epsilon^2}{\pi^2\wt{W}^2}\frac{2^{3j-6}}{(2^J-1)}$.}
		\STATE{Define a candidate set $S_j = \left\{\frac{2k\pi-\arg Z_{\delta_j}(t_j)}{2^j} \right\}_{k=0,\ldots,2^j-1}$.}
		\STATE{$\theta_j = \argmin_{\theta\in S_j}\ntp{\theta - \theta_{j-1}}$. }	
		\ENDFOR
        \STATE{Add a proper integer multiple of $2\pi$ to $\theta_{J-1}$ such that $|\theta_{J-1}|\leq\pi$.}
		\STATE{\textbf{Output:} $\wt{W}\theta_{J-1}$ as an approximation to $\omega$. }
	\end{algorithmic}
\end{algorithm}

In the following theorem, we analyze the performance of the above algorithm for a fixed set of values for $\delta$ in each iteration. We will then optimize these values to achieve the $\Or(\epsilon^{-1})$ scaling in Corollary~\ref{cor:heisenberg_limit_single_frequency}.
\begin{thm}
\label{thm:single_frequency_est_unsimplified}
    Suppose that $|\omega|<W$ is known in advance and that we have access to a signal $Z_{\delta}(t)$ such that 
    \begin{itemize}
        \item $|Z_{\delta}(t)|=1$,
        \item $|Z_{\delta}(t)-e^{-i(\omega t+f(t))}|\leq \eta$ with probability at least $1-\delta$, where $\sup_t|f(t)|\le C_f <{\pi/3}$,
        \item $2\arcsin \frac{\eta}{2} + C_f\le \frac{\pi}{3}$,
        \item generating $Z_{\delta}(t)$ requires evolution time $C_Z t(\log(\delta^{-1})+1)$,
    \end{itemize}
    then we can produce an estimate $\hat{\omega}$ such that
    \begin{equation}
    \label{eq:MSE_bound_unsimplified}
        \mathbb{E}[|\hat{\omega}-\omega|^2]\leq \sum_{j=0}^{J-1}E_j^2 \delta_j + \epsilon^2/4, 
    \end{equation}
    with total evolution time at most
    \begin{equation}
    \label{eq:total_evolution_time_unsimplified}
        \frac{\pi C_Z}{\wt{W}}\sum_{j=0}^{J-1}2^j(\log(\delta_j^{-1})+1),
    \end{equation}
    where 
    \begin{equation}
    \label{eq:choice_of_params_single_frequency_est}
        E_0=2\pi,\quad  E_j=4\pi \wt{W}/(3\cdot 2^j)\ \forall j\geq 1,\quad J = \lceil \log_2(4\pi \wt{W}/(3\epsilon))\rceil,\quad \wt{W} = 3W/\pi,
    \end{equation}
    and each $\delta_j\in (0,1]$ is arbitrarily chosen.
\end{thm}
\begin{proof}
Denote $\omega/\wt{W}=\frac{\omega\pi}{3W}$ by $\tilde{\omega}$, then $|\tilde{\omega}|<\pi/3$. 
We proceed by choosing a sequence of $t_j$, $j=0, 1, \ldots, J-1$ and refining the estimation of $\tilde{\omega}$ progressively. In each iteration we generate a signal $Z_{\delta_j}(t_j)$ for arbitrarily chosen $\delta_j$. 

First, let $t_0=1/\wt{W}$, then with probability at least $1-\delta_0$,
\[
|Z_{\delta_0}(t_0)-e^{-i(\tilde{\omega} +f(t_0))}|\leq \eta,
\]
which yields
\[
|\tilde{\omega}-(-\arg Z_{\delta_0}(t_0))|\le 2\arcsin \frac{\eta}{2} + C_f<\pi/3 \mod 2\pi.
\]
Thus $\tilde{\omega}\in (-\arg Z_{\delta_0}(t_0)-\pi/3+2k\pi, -\arg Z_{\delta_0}(t_0)+\pi/3+2k\pi)$ for some integer $k$. Let $\theta_{-1} = 0$ and $S_0 = \{-\arg Z_{\delta_0}(t_0)\}$, then by choosing $\theta_0 = \argmin_{\theta\in S_0}|\theta-\theta_{-1}|_{2\pi}=-\arg Z_{\delta_0}(t_0)$, we obtain
\[
\tilde{\omega}\in (\theta_0-\frac{\pi}{3}, \theta_0+\frac{\pi}{3}),~ \mod 2\pi.
\]
Here $\ntp{\theta}$, is defined to be the minimum distance to $0$ modulo $2\pi$, i.e., $\ntp{\theta} = \pi - |(\theta \mod2\pi)-\pi|$. 

At step $j$, we set $t_j = 2t_{j-1}$, $S_j = \left\{\frac{2k\pi-\arg Z_{\delta_j}(t_j)}{2^j} \right\}_{k=0,\ldots,2^j-1}$, and $\theta_j = \argmin_{\theta\in S_j}\ntp{\theta - \theta_{j-1}}$. 
Now we are ready to prove that if
\begin{equation}
\label{eq:jth_iteration_correct}
|Z_{\delta_{j'}}(t_{j'})-e^{-i(\tilde{\omega}2^{j'} +f(t_{j'}))}|\leq \eta,
\end{equation}
for all $0\leq j'\leq j$, then
\begin{equation}
\label{eq:est}
\tilde{\omega}\in (\theta_j-\frac{\pi}{3\cdot2^j}, \theta_j+\frac{\pi}{3\cdot2^j}),~ \mod 2\pi,
\end{equation}
for all $j$ by induction. The case $j=0$ is already proved. Suppose \eqref{eq:est} holds for $j-1$. 
Because of \ref{eq:jth_iteration_correct}, we have
\[
|\tilde{\omega}2^j-(-\arg Z_{\delta_j}(t_j))|\le 2\arcsin \frac{\eta}{2} + C_f<\pi/3 \mod 2\pi.
\]
Thus
\[
\tilde{\omega}\in I_k := \left(\frac{2k\pi -\arg Z_{\delta_j}(t_j)-\pi/3}{2^j}, \frac{2k\pi -\arg Z_{\delta_j}(t_j)+\pi/3}{2^j}\right) \mod 2\pi,
\]
for some $k=0,1, \ldots 2^{j}-1$. Notice that $\dist(I_k, I_{k'})\ge \frac{\pi}{2^{j-1}}-\frac{\pi}{3\cdot2^{j-1}} = \frac{\pi}{3\cdot2^{j-2}}$ for any $k\not=k'$, and that the length of the previous estimation $(\theta_{j-1}-\frac{\pi}{3\cdot2^{j-1}}, \theta_{j-1}+\frac{\pi}{3\cdot2^{j-1}})$ is exactly $\frac{\pi}{3\cdot2^{j-2}}$, we can ensure that only one $k^*$ satisfies $I_k\cap(\theta_{j-1}-\frac{\pi}{3\cdot2^{j-1}}, \theta_{j-1}+\frac{\pi}{3\cdot2^{j-1}})\not=\varnothing, \mod 2\pi$. Moreover, the corresponding $k^*$ satisfies 
$\frac{2k^*\pi-\arg Z_{\delta_j}(t_j)}{2^j} = \argmin_{\theta\in S_j}|\theta-\theta_{j-1}|_{2\pi}$, since
\[\left|\frac{2k^*\pi-\arg Z_{\delta_j}(t_j)}{2^j}-\theta_{j-1}\right|_{2\pi}\le \left|\frac{2k^*\pi-\arg Z_{\delta_j}(t_j)}{2^j}-\tilde{\omega}\right|_{2\pi} + \left|\tilde{\omega}-\theta_{j-1}\right|_{2\pi} < \frac{\pi}{3\cdot2^j} + \frac{\pi}{3\cdot2^{j-1}} = \frac{\pi}{2^j},\]
and 
\[
\left|\frac{2k\pi-\arg Z_{\delta_j}(t_j)}{2^j}-\theta_{j-1}\right|_{2\pi}\ge \frac{\pi}{2^{j-1}} - \left|\frac{2k^*\pi-\arg Z_{\delta_j}(t_j)}{2^j}-\theta_{j-1}\right|_{2\pi} > \frac{\pi}{2^{j-1}} - \frac{\pi}{2^j} = \frac{\pi}{2^j}
\]
for any $k\not=k^*$. Now we have proved \eqref{eq:est}. 

In the end, notice that \eqref{eq:est} has an ambiguity of modulus $2\pi$, we add a proper integer multiple of $2\pi$ to $ \theta_{J-1}$ such that $|\theta_{J-1}|\leq \pi$. We then choose this adjusted $\theta_{J-1}$ as our estimate for $\tilde{\omega}$, and our estimate for $\omega$ is $\wt{W}\theta_{J-1}=:\hat{\omega}$. 
From the above analysis we can see that if \eqref{eq:jth_iteration_correct} holds for $0\leq j'\leq j-1$, which means that all the iterations from $0$ to $j-1$ are successful, then by \eqref{eq:est}, $\tilde{\omega}$ is contained in $(\theta_{j-1}-\pi/(3\cdot 2^{j-1}),\theta_{j-1}+\pi/(3\cdot 2^{j-1}))+2k\pi$ for some integer $k$, and our estimate $\theta_{J-1}$ is contained in $(\theta_{j-1}-\pi/(3\cdot 2^{j-1}),\theta_{j-1}+\pi/(3\cdot 2^{j-1}))+2k'\pi$ for some integer $k'$. Since $|\tilde{\omega}|<\pi/3$, we have $((\theta_{j-1}-\pi/(3\cdot 2^{j-1}),\theta_{j-1}+\pi/(3\cdot 2^{j-1}))+2k\pi)\subset (-\pi, \pi)$, and then $((\theta_{j-1}-\pi/(3\cdot 2^{j-1}),\theta_{j-1}+\pi/(3\cdot 2^{j-1}))+2k'\pi)\cap [-\pi, \pi]=\varnothing$ if $k'\not=k$. Hence we must have $k=k'$ since $|\theta_{J-1}|\leq\pi$. Therefore the error in the normalized $\tilde{\omega}$ is at most $E_j/\wt{W}=4\pi/(3\cdot 2^j)$, for $j=1,2,\cdots,J-1$. If the very first iteration fails, the error is at most $E_0/\wt{W}=2\pi$. If all the iterations are successful, then by \eqref{eq:est} and the argument above, $|\theta_{J-1}-\tilde{\omega}|\leq \pi/(3\cdot 2^{J-1})\leq \epsilon/(2\wt{W})$. From these observations, we will compute the expected error.

We define the random variable $j_{\mathrm{fail}}$ to be the first iteration that fails, i.e., 
\begin{equation}
    |Z_{\delta_{j'}}(t_{j'})-e^{-i(\tilde{\omega}2^{j'} +f(t_{j'}))}|\leq \eta,\ \forall j'< j_{\mathrm{fail}},\quad |Z_{\delta_{j_{\mathrm{fail}}}}(t_{j_{\mathrm{fail}}})-e^{-i(\tilde{\omega}2^{j_{\mathrm{fail}}} +f(t_{j_{\mathrm{fail}}}))}|> \eta.
\end{equation}
If such a $j_{\mathrm{fail}}$ cannot be found, i.e., all iterations are successful, then we let $j_{\mathrm{fail}}=J$.
From the above analysis, conditional on $j_{\mathrm{fail}}=j<J$, the error will be at most $E_j/\wt{W}$. In other words $\mathbb{E}[|\tilde{\omega}-\theta_{J-1}|^2|j_{\mathrm{fail}}=j]\leq E_j^2/\wt{W}^2$. If $j=J$, then the error is at most $\epsilon/(2\wt{W})$. Also, we have
\begin{equation}
    \Pr[j_{\mathrm{fail}}=j] = (1-\delta_0)(1-\delta_1)\cdots (1-\delta_{j-1})\delta_j\leq \delta_j.
\end{equation}
Therefore the expected square error is
\begin{equation}
\begin{aligned}
    \mathbb{E}[|\omega-\hat{\omega}|^2] &= \wt{W}^2 \mathbb{E}[|\tilde{\omega}-\theta_{J-1}|^2] \\
    &=\wt{W}^2 \sum_{j=0}^{J} \mathbb{E}[|\tilde{\omega}-\theta_{J-1}|^2|j_{\mathrm{fail}}=j]\Pr[j_{\mathrm{fail}}=j] \\
    &\leq \sum_{j=0}^{J-1} E_j^2 \delta_j + \epsilon^2/4.
\end{aligned}
\end{equation}
This proves \eqref{eq:MSE_bound_unsimplified}. Generation of each $Z_{\delta_j}(t_j)$ requires an evolution time of $C_Z t_j(\log(\delta_j^{-1})+1)$, and hence we have total evolution time \eqref{eq:total_evolution_time_unsimplified} by adding them up.  
\end{proof}

In the theorem above, we have left a great deal of flexibility in choosing $\delta_j$. Below, we will try to answer that if we want the MSE to satisfy $\mathbb{E}[|\hat{\omega}-\omega|^2]\leq \epsilon^2$, how we should choose the $\delta_j$ to minimize the total evolution time required. We first state our result:
\begin{cor}
\label{cor:heisenberg_limit_single_frequency}
    Suppose that $|\omega|<W$ is known in advance and that we have access to a signal $Z_{\delta}(t)$ such that 
    \begin{itemize}
        \item $|Z_{\delta}(t)|=1$,
        \item $|Z_{\delta}(t)-e^{-i(\omega t+f(t))}|\leq \eta$ with probability at least $1-\delta$, where $\sup_t|f(t)|\le C_f <{\pi/3}$,
        \item $2\arcsin \frac{\eta}{2} + C_f\le \frac{\pi}{3}$,
        \item generating $Z_{\delta}(t)$ requires evolution time $C_Z t(\log(\delta^{-1})+1)$,
    \end{itemize}
    then we can produce an estimate $\hat{\omega}$ such that $\mathbb{E}[|\hat{\omega}-\omega|^2]\leq \epsilon^2$, with total evolution time at most $\Or(C_Z \epsilon^{-1})$.
\end{cor}

\begin{proof}
    By \eqref{eq:MSE_bound_unsimplified} and \eqref{eq:total_evolution_time_unsimplified}, we essentially need to solve the following optimization problem to get the optimal $\{\delta_j\}$:
    \begin{equation}
        \label{eq:opt_problem_for_choosing_delta_j}
        \begin{aligned}
            &\underset{\{\delta_j\}}{\text{minimize}}\  \sum_{j=0}^{J-1}2^j \log(\delta_j^{-1}) \\
            &\text{subject to}\  \sum_{j=0}^{J-1}E_j^2 \delta_j \leq \frac{3}{4}\epsilon^2.
        \end{aligned}
    \end{equation}
    This optimization problem can be easily solved using the concavity of the logarithmic function, and the optimal $\delta_j$ is
    \begin{equation}
        \label{eq:delta_j_opt}
        \delta_j = \frac{3\epsilon^2}{4E_j^2}\frac{2^j}{(2^J-1)}.
    \end{equation}
    
    Using this choice of $\delta_j$, we can then compute the total evolution time required through \ref{eq:total_evolution_time_unsimplified}.
    \begin{equation}
    \label{eq:total_evolution_time_decompose}
        \frac{\pi C_Z}{\wt{W}}\sum_{j=0}^{J-1}2^j(\log(\delta_j^{-1})+1) = \underbrace{\frac{\pi C_Z}{\wt{W}}\sum_{j=0}^{J-1}2^j\log(\delta_j^{-1})}_{\mathrm{(I)}} + \underbrace{\frac{\pi C_Z}{\wt{W}}(2^J-1)}_{\mathrm{(II)}}.
    \end{equation}
    For term $\mathrm{(II)}$, we have 
    \begin{equation}
    \label{eq:bounding_II}
        \frac{\pi C_Z}{\wt{W}}(2^J-1)<\pi C_Z \frac{8\pi}{3\epsilon}
    \end{equation}
    by our choice of $J$ given in \ref{eq:choice_of_params_single_frequency_est}.
    For $\mathrm{(I)}$, using our expression for $\delta_j$ and the expression for $E_j$ in \eqref{eq:choice_of_params_single_frequency_est}, we have
    \begin{equation}
    \label{eq:bounding_I}
    \begin{aligned}
        \mathrm{(I)} &= \frac{\pi C_Z}{\wt{W}}\sum_{j=0}^{J-1}2^j\log\left(\frac{4E_j^2}{3\epsilon^2}\frac{2^J-1}{2^j}\right) \\
        &= \frac{\pi C_Z}{\wt{W}}\log\left(\frac{16\pi^2 \wt{W}^2}{3\epsilon^2}(2^J-1)\right) + \frac{\pi C_Z}{\wt{W}}\sum_{j=1}^{J-1}2^j\log\left(\frac{64\pi^2 \wt{W}^2}{27\epsilon^2}\frac{2^J-1}{2^j}\frac{1}{4^j}\right) \\
        &< \frac{\pi C_Z}{\wt{W}}\log\left(\frac{64\pi^3 \wt{W}^3}{9\epsilon^3}\right) + \frac{\pi C_Z}{\wt{W}}\sum_{j=1}^{J-1}2^j\log\left(\frac{4}{3}8^{J-j}\right) \\
        &= \frac{\pi C_Z}{\wt{W}}\log\left(\frac{64\pi^3 \wt{W}^3}{9\epsilon^3}\right) + \frac{\pi C_Z}{\wt{W}}\log\left(\frac{4}{3}\right)(2^J-2) + \frac{\pi C_Z}{\wt{W}}\log(8)(2^{J+2}-2J-2) \\
        &\leq \Or(\epsilon) + \Or(C_Z \epsilon^{-1}).
    \end{aligned}
    \end{equation}
    In the last line, we have used the fact that $\epsilon\leq W$ (as otherwise we can simply estimate $\omega$ by $0$) to bound the first term on the second-to-last line. Combining \eqref{eq:total_evolution_time_decompose}, \eqref{eq:bounding_II}, and \eqref{eq:bounding_I}, we can see that the total evolution time of the entire procedure is $\Or(C_Z \epsilon^{-1})$.
\end{proof}

\section{Learning an anharmonic oscillator}
\label{sec:learning_an_anharmonic_oscillator_appendix}

The basic building block of our algorithm is a method to learn a single anharmonic oscillator of the form
\begin{equation}
    H_{\mathrm{AHO}} = \omega b^{\dagger}b + \frac{\xi}{2}n(n-1),
\end{equation}
where $n=b^{\dagger}b$.
We will then outline the experiments we run to learn the coefficients $\omega$ and $\xi$ from this Hamiltonian. We first start with a coherent state
\begin{equation}
    \ket{\alpha} = e^{-|\alpha|^2/2}\sum_{k=0}^{\infty}\frac{\alpha^k}{\sqrt{k!}}\ket{k}.
\end{equation}
We then let the system evolve under the Hamiltonian $H_{\mathrm{AHO}}$ for time $t$, and obtain the quantum state
\begin{equation}
    e^{-i H_{\mathrm{AHO}} t}\ket{\alpha} = e^{-|\alpha|^2/2}\sum_{k=0}^{\infty}\frac{\alpha^k}{\sqrt{k!}}e^{-i\omega k t-i\frac{\xi}{2}k(k-1)t}\ket{k}.
\end{equation}
In the end, we perform POVM measurement in the eigenbasis of either $X=(b+b^{\dagger})/\sqrt{2}$ or $P=i(b^{\dagger}-b)/\sqrt{2}$, and by taking average we obtain $\braket{X}_{\alpha,t}$ and $\braket{P}_{\alpha,t}$, where $\braket{\cdot}_{\alpha,t}$ means taking expectation with respect to the state $e^{-iH_{\mathrm{AHO}} t}\ket{\alpha}$. With these, we can then obtain the expectation value of $b$ through
\begin{equation}
    \braket{b}_{\alpha,t}=\frac{1}{\sqrt{2}}\left(\braket{X}_{\alpha,t} + i\braket{P}_{\alpha,t}\right).
\end{equation}
The expectation values $\braket{b}_{\alpha,t}$ for a certain set of $\alpha$ and $t$ will enable us to estimate $\omega$ and $\xi$, and we will demonstrate this below. First we can compute $b e^{-i H_{\mathrm{AHO}} t}\ket{\alpha}$ to be 
\begin{equation}
    \begin{aligned}
    be^{-iH_{\mathrm{AHO}}t}\ket{\alpha} &= e^{-|\alpha|^2/2}\sum_{k=1}^{\infty}\frac{\alpha^k e^{-i\omega k t}e^{-i\frac{\xi}{2} k(k-1)t}}{\sqrt{(k-1)!}}\ket{k-1} \\
    &= e^{-|\alpha|^2/2}\sum_{k=0}^{\infty}\frac{\alpha^{k+1} e^{-i\omega (k+1) t}e^{-i\frac{\xi}{2} k(k+1)t}}{\sqrt{k!}}\ket{k}.
    \end{aligned}
\end{equation}
This yields a closed-form expression for $\braket{b}_{\alpha,t}$:
\begin{equation}
\begin{aligned}
    \braket{b}_{\alpha,t} &= e^{-|\alpha|^2}\sum_k \frac{\alpha |\alpha|^{2k} e^{-i\omega t}e^{-i\xi kt}}{k!} \\
    &= \alpha e^{-|\alpha|^2} e^{-i\omega t} e^{|\alpha|^2 e^{-i\xi t}}.
\end{aligned}
\end{equation}
Now we are ready to estimate $\omega$ and $\xi$ with the help of Corollary~\ref{cor:heisenberg_limit_single_frequency}. To estimate $\omega$, we define 
\[
Z(t) = \frac{\braket{b}_{\alpha,t}}{\left|\braket{b}_{\alpha,t}\right|} = e^{-i(\omega t + |\alpha|^2\sin(\xi t))},
\]
then $Z(t) = e^{-i(\omega t + f(t))}$, where $f(t) = |\alpha|^2\sin(\xi t)$. Therefore, $\sup_{t}|f(t)|\leq|\alpha|^2$. The exact value of $Z(t)$ is, however, inaccessible in practice, and we need to find an approximation $Z_\delta(t)$ such that $|Z_\delta(t)-Z(t)|\leq \eta$ with probability at least $\delta$ if we want to utilize Corollary~\ref{cor:heisenberg_limit_single_frequency}. In the following, we decompose the approximation error into three parts and analyze them separately.

\textbf{Truncation error.}
The first part of the approximation error comes from the truncation of the observables. In our protocol, we truncate the observables up to a threshold $M$, which means rather than estimating the $\braket{X}_{\alpha,t}$ and $\braket{P}_{\alpha,t}$ we estimate $\braket{X\one{|X|\leq M}}_{\alpha,t}$ and $\braket{P\one{|P|\leq M}}_{\alpha,t}$. 
Here $\one{|X|\leq M}$ and $\one{|P|\leq M}$ are defined to be
\begin{equation*}
    \one{|X|\leq M} = \int_{-M}^M\ket{x}\bra{x}\dd x,\quad \one{|P|\leq M} = \int_{-M}^M\ket{p}\bra{p}\dd p.
\end{equation*}
This is necessary for the robustness of our protocol. With the original unbounded observables $X$ and $P$, any small error in the quantum state can potentially be infinitely magnified in the expectation value. The use of bounded observables will ensure that this does not happen. 

In the following, we will ensure that the error introduced by this truncation is acceptable for our protocol.
From Chebyshev's inequality, one has
\begin{equation}
    \mathbb{P}(|X|\ge M)\le \frac{\braket{X^2}_{\alpha,t}}{M^2} \le \frac{2\braket{b^\dagger b}_{\alpha, t}+1}{M^2} = \frac{2|\alpha|^2+1}{M^2},
\end{equation}
where we have used the fact that $\braket{b^\dagger b}_{\alpha, t}=|\alpha|^2$. Then, by Cauchy-Schwarz inequality, 
\begin{equation}
\begin{aligned}
\left|\braket{X}_{\alpha,t}-\braket{X\one{|X|\le M}}_{\alpha,t}\right|&= |\braket{X\one{|X|>M}}_{\alpha,t}|\le \sqrt{\braket{X^2}_{\alpha,t}}\sqrt{\braket{\one{|X|>M}^2}_{\alpha,t}}\\
&\le \sqrt{2\braket{b^\dagger b}_{\alpha,t}+1}\sqrt{\mathbb{P}(|X|\ge M)}\le \frac{2|\alpha|^2+1}{{M}}.
\end{aligned}
\end{equation}
Similarly, one has
\begin{equation}
    \mathbb{P}(|P|\ge M) \le \frac{\braket{P^2}_{\alpha,t}}{M^2} \le \frac{2\braket{b^\dagger b}_{\alpha,t}+1}{M^2} = \frac{2|\alpha|^2+1}{M^2},
\end{equation}
and
\begin{equation}
\begin{aligned}
\left|\braket{P}_{\alpha,t}-\braket{P\one{|P|\le M}}_{\alpha,t}\right|\le\frac{2|\alpha|^2+1}{{M}}.
\end{aligned}
\end{equation}
Combining the error bounds for $X$ and $P$ truncations, we will have an error bound for the truncated $b$ operator.
Let 
\begin{equation}
    Z_{M}(t)=\frac{1}{\sqrt{2}}\left(\braket{X\one{|X|\le M}}_{\alpha,t} + i\braket{P\one{|P|\le M}}_{\alpha,t}\right),
\end{equation}
then
\begin{equation}\label{eq:tce}
\begin{aligned}
\abs{Z_{M}(t)-\braket{b}_{\alpha,t}} &= \sqrt{\abs{\Re Z_{M}(t) - \Re \braket{b}_{\alpha,t}}^2 + \abs{\Im Z_{M}(t) - \Im \braket{b}_{\alpha,t}}^2} \\
&= \sqrt{\half\left(\abs{\braket{X}_{\alpha,t}-\braket{X\one{|X|\le M}}_{\alpha,t}}^2 + \abs{\braket{P}_{\alpha,t}-\braket{P\one{|P|\le M}}_{\alpha,t}}^2\right)}\\
&\le \frac{2|\alpha|^2+1}{{M}}.
\end{aligned}
\end{equation}

\textbf{Simulation error.}
In practice, the final state we obtained is different from the ideal state $e^{-iH_{\mathrm{AHO}} t}\ket{\alpha}$. 
This is because in the multi-mode situation, $H_{\mathrm{AHO}}$ is the Hamiltonian of the effective dynamics, which differs from the actual dynamics by a small error. In this sense, we only \emph{simulate} the effective dynamics, thus calling this error \emph{the simulation error}.
We denote the expectation with respect to the real final state obtained by $\braket{\cdot}_{\alpha, t, r}$, where $r$ stands for the parameters used in the simulation. More precisely, as will be explained in Section~\ref{sec:single_mode_coeff_appendix}, and in particular \eqref{eq:time_evolution_transformation_appendix}, $r$ is the number of random unitaries that we insert during the time evolution. In Section~\ref{sec:deviation_from_effective_dynamics_appendix}, we will show that for any given $\eta_0>0$, there exists a choice of $r$ such that 
\begin{equation}
    \abs{\braket{O}_{\alpha, t, r}-\braket{O}_{\alpha, t}} \le \norm{O}\eta_0, 
\end{equation}
for any bounded observable $O$. In particular, for any given $\eta_0>0$, there is a choice of $r$ such that
\begin{equation}
    \abs{\braket{X\one{|X|\le M}}_{\alpha, t, r}-\braket{X\one{|X|\le M}}_{\alpha, t}} \le M\eta_0, ~\abs{\braket{P\one{|P|\le M}}_{\alpha, t, r}-\braket{P\one{|P|\le M}}_{\alpha, t}} \le M\eta_0.
\end{equation}
Define
\begin{equation}
    Z_{M,r}(t)=\frac{1}{\sqrt{2}}\left(\braket{X\one{|X|\le M}}_{\alpha,t,r} + i\braket{P\one{|P|\le M}}_{\alpha,t,r}\right),
\end{equation}
then 
\begin{equation}
\begin{aligned}
&\abs{Z_{M,r}(t)-Z_{M}(t)} \\
&= \sqrt{\half\left(\abs{\braket{X\one{|X|\le M}}_{\alpha,t,r}-\braket{X\one{|X|\le M}}_{\alpha,t}}^2 + \abs{\braket{P\one{|P|\le M}}_{\alpha,t,r}-\braket{P\one{|P|\le M}}_{\alpha,t}}^2\right)}\\
&\le M\eta_0.
\end{aligned}
\end{equation}

\textbf{Statistical error.} In practice, homodyne measurement generates samples corresponding to the quadrature operator $X$. By discarding the samples with norm larger than $M$, we obtain samples $\hat{x}_1,\hat{x}_2,\cdots,\hat{x}_L$ corresponding to $X\one{|X|\le M}$. We then approximate $\braket{X\one{|X|\le M}}_{\alpha, t, r}$ through the average $\bar{x}=(\hat{x}_1+\hat{x}_2+\cdots+\hat{x}_L)/L$. Similarly, we can generate $\hat{p}_1,\hat{p}_2,\cdots,\hat{p}_L$ corresponding to $P\one{|P|\le M}$, and use $\bar{p} = (\hat{p}_1+\hat{p}_2+\cdots+\hat{p}_L)/L$ to approximate $\braket{P\one{|P|\le M}}_{\alpha, t, r}$. It is clear that $\hat{x}$ and $\hat{p}$ are unbiased estimates for $\braket{X\one{|X|\le M}}_{\alpha, t, r}$ and $\braket{P\one{|P|\le M}}_{\alpha, t, r}$. Define
\begin{equation}
    \bar{Z} = \frac{1}{\sqrt{2}}(\bar{x}+i\bar{p}),
\end{equation}
then 
\begin{equation}
\begin{aligned}
&\abs{\bar{Z}-Z_{M,r}(t)}= \sqrt{\half\left(\abs{\bar{x}-\braket{X\one{|X|\le M}}_{\alpha,t,r}}^2 + \abs{\bar{p}-\braket{P\one{|P|\le M}}_{\alpha,t,r}}^2\right)}\\
&\le \max\left\{\abs{\bar{x}-\braket{X\one{|X|\le M}}_{\alpha,t,r}}, \abs{\bar{p}-\braket{P\one{|P|\le M}}_{\alpha,t,r}}\right\}.
\end{aligned}
\end{equation}
Thus by the union bound and Hoeffding's inequality, we have
\begin{equation}
\begin{aligned}
\mathbb{P}\left(|\bar{Z}-Z_{M,r}(t)|\ge {\eta_1}\right)&\le \mathbb{P}\left(\abs{\bar{x}-\braket{X\one{|X|\le M}}_{\alpha,t,r}}\ge \eta_1\right) + \mathbb{P}\left(\abs{\bar{p}-\braket{P\one{|P|\le M}}_{\alpha,t,r}}\ge \eta_1\right) \\
&\le 2e^{-\frac{L\eta_1^2}{2M^2}}+2e^{-\frac{L\eta_1^2}{2M^2}}= 4e^{-\frac{L\eta_1^2}{2M^2}}.
\end{aligned}
\end{equation}
Putting the three types of error together, we have
\begin{equation}
\begin{aligned}\label{eq:approxerr}
\abs{\bar{Z}-\braket{b}_{\alpha,t}}&\le \abs{\bar{Z}-Z_{M,r}(t)}+\abs{Z_{M,r}(t)-Z_{M}(t)}+\abs{Z_{M}(t)-\braket{b}_{\alpha,t}}\\
&\le \eta_1+M\eta_0+\frac{2|\alpha|^2+1}{{M}},
\end{aligned}
\end{equation}
with probability at least $1-4e^{-\frac{L\eta_1^2}{2M^2}}$. Define
\begin{equation}
Z_\delta(t) = \frac{\bar{Z}}{\abs{\bar{Z}}},
\end{equation}
then
\begin{equation}
\begin{aligned}
\abs{Z_\delta(t)-Z(t)} &= \abs{\frac{\bar{Z}}{\abs{\bar{Z}}}-\frac{\braket{b}_{\alpha,t}}{\abs{\braket{b}_{\alpha,t}}}}\le \frac{2\abs{\bar{Z}-\braket{b}_{\alpha,t}}}{\abs{\braket{b}_{\alpha,t}}}\\
&=\frac{2\abs{\bar{Z}-\braket{b}_{\alpha,t}}}{|\alpha|e^{|\alpha|^2(\cos(\xi t)-1)}}\le 2|\alpha|^{-1}e^{2|\alpha|^2}\abs{\bar{Z}-\braket{b}_{\alpha,t}}.\\
\end{aligned}
\end{equation}
Hence,
\begin{equation}
\abs{Z_\delta(t)-Z(t)}\le 2|\alpha|^{-1}e^{2|\alpha|^2}\left(\eta_1+M\eta_0+\frac{2|\alpha|^2+1}{{M}}\right),
\end{equation}
with probability at least $1-4e^{-\frac{L\eta_1^2}{2M^2}}$. In order for the condition $2\arcsin \frac{\eta}{2} + C_f\le \frac{\pi}{3}$ in Theorem~\ref{thm:single_frequency_est_unsimplified} to hold, we need
\begin{equation}\label{eq:approxbound}
    2\arcsin \left(|\alpha|^{-1}e^{2|\alpha|^2}\left(\eta_1+M\eta_0+\frac{2|\alpha|^2+1}{{M}}\right)\right) + |\alpha|^2 \le \frac{\pi}{3}.
\end{equation}\label{eq:probbound}
In order for $1-4e^{-\frac{L\eta_1^2}{2M^2}}\ge1-\delta$ to hold, we need
\begin{equation}
    L\ge \frac{2M^2}{\eta_1^2}\log\frac{4}{\delta}.
\end{equation}
In conclusion, we have constructed a signal $Z_\delta(t)$ to estimate the parameter $\omega$ that satisfies the conditions required by Corollary~\ref{cor:heisenberg_limit_single_frequency}. 
\begin{lem}\label{lem:zomega}
Define $Z_\delta(t) = \bar{Z}/\abs{\bar{Z}}$, where $\bar{Z} = \frac{1}{\sqrt{2}}(\bar{x}+i\bar{p})$, and $(\bar{x}, \bar{p})$ are the average values computed from $L$ measurement results each for $\braket{X\one{|X|\le M}}_{\alpha, t, r}$ and $\braket{P\one{|P|\le M}}_{\alpha, t, r}$, respectively. Here $\braket{\cdot}_{\alpha, t, r}$ denotes the expectation with respect to the real final state obtained by a simulation using $r$ randomly inserted unitaries, which is an approximation of the state $e^{-iH_{\mathrm{AHO}} t}\ket{\alpha}$ satisfying $\abs{\braket{O}_{\alpha, t, r}-\braket{O}_{\alpha, t}} \le \norm{O}\eta_0$ for any bounded operator $O$. Then $Z_\delta(t)$ satisfies the conditions of Corollary~\ref{cor:heisenberg_limit_single_frequency} for the estimation of $\omega$ if
\begin{equation}\label{eq:params}
\begin{aligned}
&|\alpha|^2<\frac{\pi}{3},~ M > \frac{e^{2|\alpha|^2}(2|\alpha|^2+1)}{|\alpha|\sin(\pi/6-|\alpha|^2/2)}, \\
&\eta_0<\frac{1}{M}\left(|\alpha|e^{-2|\alpha|^2}\sin\left(\frac{\pi}{6}-\frac{|\alpha|^2}{2}\right)-\frac{2|\alpha|^2+1}{{M}}\right),\\
&\eta_1\le |\alpha|e^{-2|\alpha|^2}\sin\left(\frac{\pi}{6}-\frac{|\alpha|^2}{2}\right)-\frac{2|\alpha|^2+1}{{M}} -M\eta_0,\\
&L\ge \frac{2M^2}{\eta_1^2}\log\frac{4}{\delta}.
\end{aligned}
\end{equation}
As a result, $\alpha$, $M$, $\eta_0$ and $\eta_1$ can be chosen as $\mathcal{O}(1)$ constants and the total runtime needed in producing $Z_\delta(t)$ is $\mathcal{O}(t(\log(1/\delta)+1))$.
\end{lem}
\begin{rem}
When choosing the parameters in practice, one can follow the order in \eqref{eq:params}, i.e., first decide the value of $\alpha$, then choose $M$, $\eta_0$, $\eta_1$ and $L$ accordingly. 
\end{rem}

Next, we will build the signal $Z_\delta(t)$ for the estimation of $\xi$ with the help of the result above. In particular, when \eqref{eq:approxerr} and \eqref{eq:approxbound} hold, one can deduce that
\begin{equation}\label{eq:priorbound}
    |\bar{Z}-\braket{b}_{\alpha, t}| \le |\alpha|e^{-2|\alpha|^2}\sin\left(\frac{\pi}{6}-\frac{|\alpha|^2}{2}\right)\le \half|\alpha|e^{-2|\alpha|^2}\le\half|\braket{b}_{\alpha, t}|.
\end{equation}
We first observe that $\cos(\xi t)$ can be obtained by
\begin{equation}
    \cos(\xi t) = \frac{1}{|\alpha|^2}\log\frac{\abs{\braket{b}_{\alpha, t}}}{\alpha} + 1.
\end{equation}
Therefore, when when \eqref{eq:approxerr} and \eqref{eq:approxbound} hold, the error in the estimation of $\cos(\xi t)$ caused by using $\bar{Z}$ instead of $\braket{b}_{\alpha, t}$ is 
\begin{equation}\label{eq:errcos}
\begin{aligned}
&\abs{\left(\frac{1}{|\alpha|^2}\log\frac{\abs{\braket{b}_{\alpha, t}}}{\alpha} + 1\right) -
\left(\frac{1}{|\alpha|^2}\log\frac{\abs{\bar{Z}}}{\alpha}+1\right)}\\
=&~\frac{1}{|\alpha|^2}\abs{\log\frac{\abs{\bar{Z}}}{\abs{\braket{b}_{\alpha, t}}}}=\frac{1}{|\alpha|^2}\abs{\log\left(1+\frac{\abs{\bar{Z}}-\abs{\braket{b}_{\alpha, t}}}{\abs{\braket{b}_{\alpha, t}}}\right)} \\
\le&~\frac{2\log2}{|\alpha|^2}\abs{\frac{\abs{\bar{Z}}-\abs{\braket{b}_{\alpha, t}}}{\abs{\braket{b}_{\alpha, t}}}}\\
\le&~2\log2|\alpha|^{-3}e^{2|\alpha|^2}\left(\eta_1+M\eta_0+\frac{2|\alpha|^2+1}{{M}}\right),
\end{aligned}
\end{equation}
where we have used the concavity of $\log x$ and \eqref{eq:priorbound} in the third line. For estimating $\sin(\xi t)$, we use two different values for $\alpha$. The ratio between $\braket{b}_{\alpha_1,t}$ and $\braket{b}_{\alpha_2,t}$ is
\begin{equation}
    \frac{\braket{b}_{\alpha_1,t}}{\braket{b}_{\alpha_2,t}} = \frac{\alpha_1}{\alpha_2} e^{(|\alpha_2|^2-|\alpha_1|^2)(1-e^{-i\xi t})}. 
\end{equation}
Let $Z_{\alpha_1, \alpha_2} = \frac{\braket{b}_{\alpha_1,t}/\braket{b}_{\alpha_2,t}}{\abs{\braket{b}_{\alpha_1,t}/\braket{b}_{\alpha_2,t}}}$ 
and $\beta = \abs{\alpha_2}^2-\abs{\alpha_1}^2$. 
Assume that $\abs{\beta}<\frac{\pi}{2}$, then
\begin{equation}
\begin{aligned}
    \sin(\xi t) &= \frac{1}{\beta}\arcsin\left(\Im Z_{\alpha_1, \alpha_2}\right),
\end{aligned}
\end{equation}
Now we analyze the error in the estimate of $\sin(\xi t)$ caused by approximation. We assume that \eqref{eq:approxbound} holds for both $\alpha_1$ and $\alpha_2$, and we condition on the event that 
\begin{equation}
\begin{aligned}
    &\abs{\bar{Z}_1-\braket{b}_{\alpha_1, t}} \le \eta_1+M\eta_0+\frac{2|\alpha_1|^2+1}{{M}},\\
    &\abs{\bar{Z}_2-\braket{b}_{\alpha_2, t}} \le \eta_1+M\eta_0+\frac{2|\alpha_2|^2+1}{{M}}.
\end{aligned} 
\end{equation}
Then
\begin{equation}
\begin{aligned}
&\abs{\Im \frac{\bar{Z}_1/\bar{Z}_2}{\abs{\bar{Z}_1/\bar{Z}_2}}-\Im \frac{\braket{b}_{\alpha_1,t}/\braket{b}_{\alpha_2,t}}{\abs{\braket{b}_{\alpha_1,t}/\braket{b}_{\alpha_2,t}}}}\\
\le&~\abs{ \frac{\bar{Z}_1/\bar{Z}_2}{\abs{\bar{Z}_1/\bar{Z}_2}}- \frac{\braket{b}_{\alpha_1,t}/\braket{b}_{\alpha_2,t}}{\abs{\braket{b}_{\alpha_1,t}/\braket{b}_{\alpha_2,t}}}}\\
\le&~\frac{2\abs{\bar{Z}_1/\bar{Z}_2-\braket{b}_{\alpha_1,t}/\braket{b}_{\alpha_2,t}}}{\abs{\braket{b}_{\alpha_1,t}/\braket{b}_{\alpha_2,t}}}\\
\le&~2\frac{\abs{\braket{b}_{\alpha_2, t}}}{\abs{\braket{b}_{\alpha_1, t}}}\left[\frac{\abs{\bar{Z}_1-\braket{b}_{\alpha_1,t}}}{\abs{\bar{Z}_2}}+\frac{\abs{\braket{b}_{\alpha_1,t}}}{\abs{\braket{b}_{\alpha_2,t}}}\frac{\abs{\bar{Z}_2-\braket{b}_{\alpha_2,t}}}{\abs{\bar{Z}_2}}\right]\\
\le&~4\left[\frac{\abs{\bar{Z}_1-\braket{b}_{\alpha_1,t}}}{\abs{\braket{b}_{\alpha_1,t}}}+\frac{\abs{\bar{Z}_2-\braket{b}_{\alpha_2,t}}}{\abs{\braket{b}_{\alpha_2,t}}}\right]\\
\le&~4|\alpha_1|^{-1}e^{2|\alpha_1|^2}\left(\eta_1+M\eta_0+\frac{2|\alpha_1|^2+1}{{M}}\right)+4|\alpha_2|^{-1}e^{2|\alpha_2|^2}\left(\eta_1+M\eta_0+\frac{2|\alpha_2|^2+1}{{M}}\right).
\end{aligned}
\end{equation}
Here we have used the fact that $\abs{\bar{Z}}\ge\half\abs{\braket{b}_{\alpha_,t}}$ in the fifth line, which can be deduced from \eqref{eq:approxbound}. Now, if we further assume that $\abs{\beta}\le \frac{\pi}{3}$, then since the function $\arcsin$ is $2$-Lipschitz on $[-\sin\frac{\pi}{3}, \sin\frac{\pi}{3}]$, we have
\begin{equation}\label{eq:errsin}
\begin{aligned}
&\abs{\frac{1}{\beta}\arcsin\left(\Im \frac{\bar{Z}_1/\bar{Z}_2}{\abs{\bar{Z}_1/\bar{Z}_2}}\right)-\frac{1}{\beta}\arcsin(\Im Z_{\alpha_1, \alpha_2})}\\
\le&~\frac{2}{\beta}\abs{\Im \frac{\bar{Z}_1/\bar{Z}_2}{\abs{\bar{Z}_1/\bar{Z}_2}}-\Im Z_{\alpha_1, \alpha_2}}\\
\le&~\frac{8}{\beta}\left(|\alpha_1|^{-1}e^{2|\alpha_1|^2}\left(\eta_1+M\eta_0+\frac{2|\alpha_1|^2+1}{{M}}\right)+|\alpha_2|^{-1}e^{2|\alpha_2|^2}\left(\eta_1+M\eta_0+\frac{2|\alpha_2|^2+1}{{M}}\right)\right)
\end{aligned}
\end{equation}
Combining \eqref{eq:errcos} and \eqref{eq:errsin}, we have
\begin{equation}
\begin{aligned}
&\abs{e^{i\xi t}-\frac{\hat{c}+i\hat{s}}{\abs{\hat{c}+i\hat{s}}}} \le 4\log2|\alpha|^{-3}e^{2|\alpha|^2}\left(\eta_1+M\eta_0+\frac{2|\alpha|^2+1}{{M}}\right)\\
&+\frac{16}{\beta}\left(|\alpha_1|^{-1}e^{2|\alpha_1|^2}\left(\eta_1+M\eta_0+\frac{2|\alpha_1|^2+1}{{M}}\right)+|\alpha_2|^{-1}e^{2|\alpha_2|^2}\left(\eta_1+M\eta_0+\frac{2|\alpha_2|^2+1}{{M}}\right)\right),
\end{aligned}
\end{equation}
where $\hat{c} = \frac{1}{|\alpha|^2}\log\frac{\abs{\bar{Z}}}{\alpha} + 1$ and $\hat{s} = \frac{1}{\beta}\arcsin\left(\Im \frac{\bar{Z}_1/\bar{Z}_2}{\abs{\bar{Z}_1/\bar{Z}_2}}\right)$, 
and the condition in Corollary~\ref{cor:heisenberg_limit_single_frequency} reads
\begin{equation}
\begin{aligned}
&1\ge4\log2|\alpha|^{-3}e^{2|\alpha|^2}\left(\eta_1+M\eta_0+\frac{2|\alpha|^2+1}{{M}}\right)\\
&+\frac{16}{\beta}\left(|\alpha_1|^{-1}e^{2|\alpha_1|^2}\left(\eta_1+M\eta_0+\frac{2|\alpha_1|^2+1}{{M}}\right)+|\alpha_2|^{-1}e^{2|\alpha_2|^2}\left(\eta_1+M\eta_0+\frac{2|\alpha_2|^2+1}{{M}}\right)\right).
\end{aligned}
\end{equation}
In particular, we can take $\alpha=\alpha_1$ and obtain the following result.
\begin{lem}\label{lem:zxi}
Define $Z_\delta(t) = \frac{\hat{c}+i\hat{s}}{\abs{\hat{c}+i\hat{s}}}$, where $\hat{c} = \frac{1}{|\alpha_1|^2}\log\frac{\abs{\bar{Z}_1}}{\alpha_1} + 1$, $\hat{s} = \frac{1}{\beta}\arcsin\left(\Im \frac{\bar{Z}_1/\bar{Z}_2}{\abs{\bar{Z}_1/\bar{Z}_2}}\right)$, and $(\bar{Z}_1, \bar{Z}_2)$ are defined in the same way as in Lemma~\ref{lem:zomega} for $\alpha_1$ and $\alpha_2$, respectively. Then $Z_\delta(t)$ satisfies the conditions of Corollary~\ref{cor:heisenberg_limit_single_frequency} for the estimation of $\xi$ if
\begin{equation}\label{eq:paramsxi}
\begin{aligned}
&|\alpha_1|^2<\frac{\pi}{3},~|\alpha_2|^2<\frac{\pi}{3},~\beta := \abs{|\alpha_1|^2-|\alpha_2|^2}<\frac{\pi}{2},\\
&M > \left(4\log2|\alpha_1|^{-3}+\frac{16}{\beta}|\alpha_1|^{-1}\right)e^{2|\alpha_1|^2}\left({2|\alpha_1|^2+1}\right) + \frac{16}{\beta}|\alpha_2|^{-1}e^{2|\alpha_2|^2}\left({2|\alpha_2|^2+1}\right), \\
&\eta_0<\frac{M-\left(4\log2|\alpha_1|^{-3}+\frac{16}{\beta}|\alpha_1|^{-1}\right)e^{2|\alpha_1|^2}\left({2|\alpha_1|^2+1}\right) - \frac{16}{\beta}|\alpha_2|^{-1}e^{2|\alpha_2|^2}\left({2|\alpha_2|^2+1}\right)}{M^2\left(\left(4\log2|\alpha_1|^{-3}+\frac{16}{\beta}|\alpha_1|^{-1}\right)e^{2|\alpha_1|^2}+\frac{16}{\beta}|\alpha_2|^{-1}e^{2|\alpha_2|^2}\right)},\\
&\eta_1\le \frac{M-\left(4\log2|\alpha_1|^{-3}+\frac{16}{\beta}|\alpha_1|^{-1}\right)e^{2|\alpha_1|^2}\left({2|\alpha_1|^2+1+M^2\eta_0}\right) - \frac{16}{\beta}|\alpha_2|^{-1}e^{2|\alpha_2|^2}\left({2|\alpha_2|^2+1+M^2\eta_0}\right)}{M\left(\left(4\log2|\alpha_1|^{-3}+\frac{16}{\beta}|\alpha_1|^{-1}\right)e^{2|\alpha_1|^2}+\frac{16}{\beta}|\alpha_2|^{-1}e^{2|\alpha_2|^2}\right)},\\
&L\ge \frac{2M^2}{\eta_1^2}\log\frac{8}{\delta}.
\end{aligned}
\end{equation}
As a result, $\alpha_1$, $\alpha_2$, $M$, $\eta_0$ and $\eta_1$ can be chosen as $\mathcal{O}(1)$ constants and the total runtime needed in producing $Z_\delta(t)$ is $\mathcal{O}(t(\log(1/\delta)+1))$.
\end{lem}

\section{Learning two coupled anharmonic oscillators}
\label{sec:learning_two_coupled_anharmonic_oscillators_appendix}

In this section, we consider a system consisting of two coupled anharmonic oscillators, and the Hamiltonian is of the following form:
\begin{equation}
\label{eq:two_mode_hamiltonian_appendix}
    H = \omega_1 b_1^{\dagger}b_1 + \omega_2 b_2^{\dagger}b_2 + h_{12}b_1^{\dagger}b_2 + h_{21}b_2^{\dagger}b_1 + \frac{\xi_1}{2}n_1(n_1-1) + \frac{\xi_2}{2}n_2(n_2-1)
\end{equation}
The goal is to learn all the coefficients $\omega_{1}$, $\omega_{2}$, $\xi_{1}$, $\xi_{2}$, and $h_{12}$ ($h_{21}=h^*_{12}$).

\subsection{Single-mode coefficients}
\label{sec:single_mode_coeff_appendix}
We first focus on learning the single-mode coefficients $\omega_1$, $\omega_2$, $\xi_1$, and $\xi_2$. To do this, we will insert random unitaries during time evolution to decouple the bosonic modes from each other. In other words, the time evolution operator undergoes the following transformation 
\begin{equation}
    \label{eq:time_evolution_transformation_appendix}
    e^{-iHt} \mapsto \prod_{j=1}^{r} U_j^{\dagger}e^{-iH\tau} U_j = \prod_{j=1}^{r} e^{-iU_j^{\dagger}HU_j\tau},
\end{equation}
where the $U_j$, $j=1,2,\cdots,r$, are the random linear optics unitaries that we insert, $r=t/\tau$, and the product goes from right to left. Each $U_j$ is independently drawn from a distribution that we denote by $\mathcal{D}$. In the limit of $\tau\to 0$, the dynamics can be described by an effective Hamiltonian
\begin{equation}
    \label{eq:effective_hamiltonian_appendix}
    H_{\mathrm{effective}} = \mathbb{E}_{U\sim \mathcal{D}} U^{\dagger}HU.
\end{equation}
This can be seen by considering the Taylor expansion of the time-evolved state in a small time step:
\begin{equation}
\label{eq:first_order_approx_short_time_step_appendix}
\begin{aligned}
    \mathbb{E}_{U\sim\mathcal{D}}[e^{-iU^{\dagger}HU\tau}\rho e^{iU^{\dagger}HU\tau}] &= \rho - i\tau \mathbb{E}_{U\sim\mathcal{D}}[[U^{\dagger}HU,\rho]] + \Or(\tau^2) \\ 
    &= e^{-i\mathbb{E}_{U\sim\mathcal{D}}[U^{\dagger}HU]\tau}\rho e^{i\mathbb{E}_{U\sim\mathcal{D}}[U^{\dagger}HU]\tau} + \Or(\tau^2).
\end{aligned}
\end{equation}
The above is not a rigorous proof because the $\Or(\tau^2)$ residue is an unbounded operator. We will provide a rigorous bound of how far the actual dynamics deviate from the limiting effective dynamics with finite $\tau>0$ in Section~\ref{sec:deviation_from_effective_dynamics_appendix}.

To learn all the single mode coefficients, we let the unitary $U$ drawn from the distribution $\mathcal{D}$ be 
\begin{equation}
    U = e^{-i\theta b_1^{\dagger}b_1},\quad \theta\sim\mathcal{U}([0,2\pi]).
\end{equation}
Here $\mathcal{U}([0,2\pi])$ is the uniform distribution over $[0,2\pi]$.
We can then compute the effective Hamiltonian
\begin{equation}
    H_{\mathrm{effective}} = \frac{1}{2\pi}\int_0^{2\pi} e^{i\theta b_1^{\dagger}b_1}He^{-i\theta b_1^{\dagger}b_1}\dd \theta = \omega_1 b_1^{\dagger}b_1 + \omega_2 b_2^{\dagger}b_2 + \frac{\xi_1}{2}n_1(n_1-1) + \frac{\xi_2}{2}n_2(n_2-1).
\end{equation}
In other words, the coupling term $h_{12}b_1^{\dagger}b_2 + h_{21}b_2^{\dagger}b_1$ is cancelled in the process, due to the equality
\begin{equation}
    \frac{1}{2\pi}\int_0^{2\pi} e^{i\theta b_1^{\dagger}b_1}b_1e^{-i\theta b_1^{\dagger}b_1} \dd\theta=\frac{1}{2\pi}\int_0^{2\pi} e^{i\theta}b_1 \dd\theta = 0.
\end{equation}
We can interpret this procedure as enforcing a particle number conservation on the first bosonic mode.

The effective Hamiltonian has the desirable feature that the two bosonic modes are no longer coupled together. Therefore we can apply the learning algorithm described in Section~\ref{sec:learning_an_anharmonic_oscillator_appendix} to learn the parameters of the two modes separately.

\subsection{The coupling coefficient}
\label{sec:coupling_coeff_appendix}

Next, we consider learning the coupling coefficient $h_{12}$. We observe that the coupling term can be transformed into a local one under a single-particle basis transformation. This is done through the following two operators
\begin{equation}
    U_{x}(\theta) = e^{i\theta (b_1^{\dagger}b_2+b_2^{\dagger}b_1)},\quad U_{y}(\theta) = e^{\theta (b_1^{\dagger}b_2-b_2^{\dagger}b_1)},
\end{equation}
which correspond to Pauli-$X$ and $Y$ rotations. They transform the annihilation operators in the following way
\begin{equation}
\label{eq:single_particle_basis_rotation_XY}
    \begin{pmatrix}
        U_{x}(\theta)b_1 U_{x}^{\dagger}(\theta) \\
        U_{x}(\theta)b_2 U_{x}^{\dagger}(\theta)
    \end{pmatrix}
    =
    \begin{pmatrix}
        \cos(\theta) & i\sin(\theta) \\
        i\sin(\theta) & \cos(\theta)
    \end{pmatrix}
    \begin{pmatrix}
        b_1 \\
        b_2
    \end{pmatrix},\ 
    \begin{pmatrix}
        U_{y}(\theta)b_1 U_{y}^{\dagger}(\theta) \\
        U_{y}(\theta)b_2 U_{y}^{\dagger}(\theta)
    \end{pmatrix}
    =
    \begin{pmatrix}
        \cos(\theta) & \sin(\theta) \\
        -\sin(\theta) & \cos(\theta)
    \end{pmatrix}
    \begin{pmatrix}
        b_1 \\
        b_2
    \end{pmatrix}.
\end{equation}
We first perform the Pauli-$Y$ rotation and define 
\begin{equation}
\label{eq:change_of_basis_Y_pi/4}
    \tilde{b}_1 = U_y(\pi/4)b_1 U_y^{\dagger}(\pi/4),\quad \tilde{b}_2 = U_y(\pi/4)b_2 U_y^{\dagger}(\pi/4).
\end{equation}
Through \eqref{eq:single_particle_basis_rotation_XY} we have
\begin{equation}
    \begin{pmatrix}
        b_1 \\
        b_2
    \end{pmatrix}
    =\frac{1}{\sqrt{2}}
    \begin{pmatrix}
        1 & -1 \\
        1 & 1
    \end{pmatrix}
    \begin{pmatrix}
        \tilde{b}_1 \\
        \tilde{b}_2
    \end{pmatrix}
\end{equation}
We will then rewrite the Hamiltonian \eqref{eq:two_mode_hamiltonian_appendix} in terms of $\tilde{b}_1$ and $\tilde{b}_2$. The quadratic part of $H$ can be written as
\begin{equation}
\label{eq:transformed_quadratic_part}
    \tilde{\omega}_1 \tilde{b}_1^{\dagger}\tilde{b}_1 + \tilde{\omega}_2 \tilde{b}_2^{\dagger}\tilde{b}_2 + \tilde{h}_{12}\tilde{b}_1^{\dagger}\tilde{b}_2 + \tilde{h}_{21}\tilde{b}_2^{\dagger}\tilde{b}_1,
\end{equation}
where
\begin{equation}
    \begin{pmatrix}
        \tilde{\omega}_1 & \tilde{h}_{12} \\
         \tilde{h}_{21} & \tilde{\omega}_2
    \end{pmatrix}
    =\frac{1}{2}
    \begin{pmatrix}
        1 & 1\\
        -1 & 1
    \end{pmatrix}
    \begin{pmatrix}
        \omega_1 & h_{12} \\
        h_{21} & \omega_2
    \end{pmatrix}
    \begin{pmatrix}
        1 & -1\\
        1 & 1
    \end{pmatrix}.
\end{equation}
In particular, we have
\begin{equation}
\label{eq:omega_tilde_expression_appendix}
    \tilde{\omega}_1 = \frac{\omega_1+\omega_2}{2}+\Re h_{12}.
\end{equation}
For the quartic part, we have 
\begin{equation}
\label{eq:transformed_quartic_part}
\begin{aligned}
    \frac{\xi_1}{2}n_1(n_1-1) &= \frac{\xi_1}{2}b_1^{\dagger}b_1^{\dagger}b_1b_1 = \sum_{ijkl=1}^2 \xi^{(1)}_{ijkl}\tilde{b}^{\dagger}_i\tilde{b}^{\dagger}_j\tilde{b}_k\tilde{b}_l, \\
    \frac{\xi_2}{2}n_2(n_2-1) &= \frac{\xi_2}{2}b_2^{\dagger}b_2^{\dagger}b_2b_2 = \sum_{ijkl=1}^2 \xi^{(2)}_{ijkl}\tilde{b}^{\dagger}_i\tilde{b}^{\dagger}_j\tilde{b}_k\tilde{b}_l.
\end{aligned}
\end{equation}
In particular
\begin{equation}
\label{eq:leading_quartic_term_coeff}
    \xi^{(1)}_{1111} = \frac{\xi_1}{4},\quad \xi^{(2)}_{1111} = \frac{\xi_2}{4}.
\end{equation}
Combining \eqref{eq:transformed_quadratic_part} and \eqref{eq:transformed_quartic_part}, the Hamiltonian $H$ can be written in terms of $\tilde{b}_1$ and $\tilde{b}_2$ as
\begin{equation}
    \label{eq:H_under_transformed_basis}
    H =  \tilde{\omega}_1 \tilde{b}_1^{\dagger}\tilde{b}_1 + \tilde{\omega}_2 \tilde{b}_2^{\dagger}\tilde{b}_2 + \tilde{h}_{12}\tilde{b}_1^{\dagger}\tilde{b}_2 + \tilde{h}_{21}\tilde{b}_2^{\dagger}\tilde{b}_1 + \sum_{ijkl=1}^2 (\xi^{(1)}_{ijkl}+\xi^{(2)}_{ijkl})\tilde{b}^{\dagger}_i\tilde{b}^{\dagger}_j\tilde{b}_k\tilde{b}_l.
\end{equation}

The above expression is much more complicated than the original expression in \eqref{eq:two_mode_hamiltonian_appendix}, but we will use random unitaries to produce a much simpler effective Hamiltonian. This time, the random unitary we use will be
\begin{equation}
    U = e^{-i\theta \tilde{b}_1^{\dagger}\tilde{b}_1},\quad \theta\sim\mathcal{U}([0,2\pi]).
\end{equation}
With the same derivation as in \eqref{eq:first_order_approx_short_time_step_appendix}, we can obtain the effective Hamiltonian as $\mathbb{E}[U^{\dagger} HU]$. Note that in conjugating with $e^{i\theta \tilde{b}_1^{\dagger}\tilde{b}_1}$, each $\tilde{b}_1$ in the Hamiltonian acquires a phase $e^{i\theta}$, and each $\tilde{b}_1^{\dagger}$ acquires a phase $e^{-i\theta}$. If in a Hamiltonian term $\tilde{b}_1$ and $\tilde{b}_1^{\dagger}$ do not appear the same number of times, then the term will acquire a phase $e^{ic\theta}$ with $c\in\{-2,-1,1,2\}$, and integrating over $\theta$ will cancel out this term. For example
\begin{equation}
    \begin{aligned}
        &\frac{1}{2\pi}\int_0^{2\pi} e^{i\theta \tilde{b}_1^{\dagger}\tilde{b}_1} \tilde{b}_1^{\dagger}\tilde{b}_2 e^{-i\theta \tilde{b}_1^{\dagger}\tilde{b}_1} \dd \theta = \frac{1}{2\pi}\int_0^{2\pi} e^{i\theta} \tilde{b}_1^{\dagger}\tilde{b}_2 \dd \theta = 0,\\
        &\frac{1}{2\pi}\int_0^{2\pi} e^{i\theta \tilde{b}_1^{\dagger}\tilde{b}_1} \tilde{b}_1^{\dagger}\tilde{b}_1^{\dagger}\tilde{b}_2\tilde{b}_2 e^{-i\theta \tilde{b}_1^{\dagger}\tilde{b}_1} \dd \theta = \frac{1}{2\pi}\int_0^{2\pi} e^{2i\theta} \tilde{b}_1^{\dagger}\tilde{b}_1^{\dagger}\tilde{b}_2\tilde{b}_2 \dd \theta = 0.
    \end{aligned}
\end{equation}
In other words, only the terms that conserve the particle number on the first bosonic mode are preserved in the effective Hamiltonian. We can then write the effective Hamiltonian as
\begin{equation}
\label{eq:effective_hamiltonian_two_mode_intermediate}
    H_{\mathrm{effective}} = \tilde{\omega}_1 \tilde{b}_1^{\dagger}\tilde{b}_1 + \tilde{\omega}_2 \tilde{b}_2^{\dagger}\tilde{b}_2 + (\xi^{(1)}_{1111}+\xi^{(2)}_{1111})\tilde{n}_1(\tilde{n}_1-1) + (A\tilde{n}_1+B\tilde{n}_2+C)\tilde{n}_2,
\end{equation}
where $\tilde{n}_1 = \tilde{b}_1^{\dagger}\tilde{b}_1$, and $\tilde{n}_2 = \tilde{b}_2^{\dagger}\tilde{b}_2$.

Recall that our goal is to learn the coupling coefficient $h_{12}$, whose real part can be derived from $\tilde{\omega}_1$, $\omega_1$, and $\omega_2$ through \eqref{eq:omega_tilde_expression_appendix}, and $\omega_1$, and $\omega_2$ can be learned using the procedure outlined in Section~\ref{sec:single_mode_coeff_appendix}. We, therefore, only need to estimate $\tilde{\omega}_1$ from the effective Hamiltonian.

To do this, we start with a product state $\ket{\alpha}\ket{0}$ on the two bosonic modes. Then we apply $U_y(\pi/4)$ to this state to get the initial state of our time evolution
\begin{equation}
    \ket{\Phi(0)} = U_y(\pi/4)\ket{\alpha}\ket{0}.
\end{equation}
This state is the tensor product of the coherent states of $\tilde{b}_1$ and $\tilde{b}_2$ because one can verify that, using \eqref{eq:change_of_basis_Y_pi/4}, 
\begin{equation}
\label{eq:coherent_state_tilde_b}
    \begin{aligned}
        \tilde{b}_1\ket{\Phi(0)}=\tilde{b}_1 U_y(\pi/4)\ket{\alpha}\ket{0} &= U_y(\pi/4) b_1\ket{\alpha}\ket{0} = \alpha U_y(\pi/4)\ket{\alpha}\ket{0} \\
        \tilde{b}_2\ket{\Phi(0)}=\tilde{b}_2 U_y(\pi/4)\ket{\alpha}\ket{0} &= U_y(\pi/4) b_2\ket{\alpha}\ket{0} = 0.
    \end{aligned}
\end{equation}
Because of the above equation, we can see that there is no particle in the bosonic mode $\tilde{b}_2$ in this state $\ket{\Phi(0)}$. As the effective Hamiltonian in \eqref{eq:effective_hamiltonian_two_mode_intermediate} conserves the particle number on both bosonic modes, the particle number on the mode $\tilde{b}_2$ will stay $0$. Consequently, any term that involves $\tilde{n}_2$ will not affect the dynamics. Therefore we can safely discard these terms and get a new effective Hamiltonian
\begin{equation}
\label{eq:two_mode_effective_hamiltonian}
    H_{\mathrm{effective}}' = \tilde{\omega}_1 \tilde{b}_1^{\dagger}\tilde{b}_1 + (\xi^{(1)}_{1111}+\xi^{(2)}_{1111})\tilde{n}_1(\tilde{n}_1-1).
\end{equation}
Note that this Hamiltonian only acts non-trivially on the bosonic mode $\tilde{b}_1$. Therefore we can use the single-mode protocol in Section~\ref{sec:learning_an_anharmonic_oscillator_appendix} to learn the coefficient $\tilde{\omega}_1$. As guaranteed in \eqref{eq:coherent_state_tilde_b}, we start from the $\alpha$-coherent state for $\tilde{b}_1$. In the time evolution, the expectation value $\braket{\tilde{b}_1}$ contains the information to determine $\tilde{\omega}_1$. The expectation value $\braket{\tilde{b}_1}$ can be extracted through homodyne measurement with two quadrature operators. 
Note that we need to convert this homodyne measurement into homodyne measurement for $b_1$ or $b_2$. This can be easily done because $\tilde{b}_1 = U_y(\pi/4)b_1 U_y^{\dagger}(\pi/4)$. We can therefore apply the unitary $U_y^{\dagger}(\pi/4)$ at the end of the time evolution and then perform homodyne measurement for $(b_1+b_1^{\dagger})/\sqrt{2}$ and $i(b_1-b_1^{\dagger})/\sqrt{2}$, which combined yields the expectation value $\braket{\tilde{b}_1}$.

Let us now briefly summarize the whole procedure. We start from a state $\ket{\alpha}\ket{0}$, apply $U_y(\pi/4)$, let the system evolve for time $t=r\tau$, while applying random $e^{-i\theta\tilde{b}_1^{\dagger}\tilde{b}_1}$ with interval $\tau$, and in the end apply $U_y^{\dagger}(\pi/4)=U_y(-\pi/4)$, after which we perform homodyne measurement for $b_1$. The quantum state right before the measurement is applied is
\begin{equation}
    U_y(-\pi/4)\prod_{j=1}^r(e^{i\theta_j\tilde{b}_1^{\dagger}\tilde{b}_1}e^{-iH\tau}e^{-i\theta_j\tilde{b}_1^{\dagger}\tilde{b}_1})U_y(\pi/4)\ket{\alpha}\ket{0},
\end{equation}
for randomly sampled $\theta_j$, $j=1,2,\cdots,r$.
Note that $e^{-i\theta_j\tilde{b}_1^{\dagger}\tilde{b}_1}=e^{-i(\theta_j/2)(n_1+n_2)}U_x(-\theta_j/2)$, and $H$ commute with $n_1+n_2$ because the particle number is conserved. We therefore have
\begin{equation}
    e^{i\theta_j\tilde{b}_1^{\dagger}\tilde{b}_1}e^{-iH\tau}e^{-i\theta_j\tilde{b}_1^{\dagger}\tilde{b}_1} = U_x(\theta/2)e^{-iH\tau}U_x(-\theta/2).
\end{equation}
Consequently we can replace all $e^{-i\theta_j\tilde{b}_1^{\dagger}\tilde{b}_1}$ with $U_x(-\theta_j/2)$. The quantum state we get in the end is, therefore
\begin{equation}
    U_y\left(-\frac{\pi}{4}\right)\prod_{j=1}^r\left(U_x\left(\frac{\theta_j}{2}\right)e^{-iH\tau}U_x\left(-\frac{\theta_j}{2}\right)\right)U_y\left(\frac{\pi}{4}\right)\ket{\alpha}\ket{0}.
\end{equation}
Note that the adjacent $U_x(-\theta_j/2)$ and $U_x(\theta_{j-1}/2)$ can be merged into $U_x(-(\theta_j-\theta_{j-1})/2)$, so that we only need to apply one $X$ rotation in each time step instead of two.



In the above procedure, we estimate $\tilde{\omega}_1$, which through \eqref{eq:omega_tilde_expression_appendix} we can estimate $\Re h_{12}$. For $\Im h_{12}$, we can instead define 
\begin{equation}
    \tilde{b}_1 = U_x(\pi/4)b_1U_x^{\dagger}(\pi/4),\quad \tilde{b}_2 = U_x(\pi/4)b_2U_x^{\dagger}(\pi/4),
\end{equation}
and then \eqref{eq:omega_tilde_expression} will become
\begin{equation}
    \tilde{\omega}_1 = \frac{\omega_1+\omega_2}{2}+\Im h_{12}.
\end{equation}
We can then change the whole procedure accordingly to estimate $\Im h_{12}$, and the corresponding state before the measurement is
\begin{equation}
    U_x\left(-\frac{\pi}{4}\right)\prod_{j=1}^r\left(U_y\left(-\frac{\theta_j}{2}\right)e^{-iH\tau}U_y\left(\frac{\theta_j}{2}\right)\right)U_x\left(\frac{\pi}{4}\right)\ket{\alpha}\ket{0}.
\end{equation}

\section{Using a divide-and-conquer approach to learn an $N$-mode system}
\label{sec:divide_and_conquer_for_N_mode_appendix}

In this section, we consider the general case, where the Hamiltonian is of the form:
\begin{equation}
\label{eq:hamiltonian_general_appendix}
    H = \sum_{\braket{i,j}} h_{ij}b_i^{\dagger}b_j + \sum_i \omega_i b_i^{\dagger}b_i + \frac{\xi_i}{2}\sum_i n_i(n_i-1).
\end{equation}
We will use a divide-and-conquer approach to learn the coefficients in this Hamiltonian. Specifically, we will insert random unitaries during time evolution to decouple the system into clusters containing one or two modes that do not interact with each other and learn the coefficients in each cluster in parallel.

We assume that the bosonic modes are arranged on a graph $\mathcal{G}=(\mathcal{V},\mathcal{E})$, where $\mathcal{V}$ is the set containing all vertices, each of which corresponds to a bosonic mode, and $\mathcal{E}$ contains all edges. $\sum_{\braket{i,j}}$ means summation over all vertices linked by an edge.

We consider decoupling the system with the help of a graph $\mathcal{L}=(\mathcal{E},\mathcal{E}_{\mathcal{L}})$ that is the link graph of $\mathcal{G}$. The set $\mathcal{E}$ is the set of all edges in $\mathcal{G}$, and $\mathcal{E}_{\mathcal{L}}$ is the set of edges of $\mathcal{L}$, which we will now define. For any two edges $e,e'\in\mathcal{E}$, we have $(e,e')\in \mathcal{E}$ if and only if they share a vertex in $\mathcal{V}$. 

Next, we color the graph $\mathcal{L}$ with the following rule: any two vertices in $\mathcal{E}$ must be colored differently if they are at most distance $2$ from each other. The number of colors needed for this coloring is at most $\chi=\operatorname{deg}(\mathcal{L})^2+1$, and such a coloring can be easily found by a greedy algorithm: we can simply color a vertex by any color that its neighbors or next-neighbors have not used, and such a color is always available because there are at most $\chi-1$ neighbors and next-neighbors. For a graph $\mathcal{G}$ with degree $D$, $\operatorname{deg}(\mathcal{L})\leq 2(D-1)$, and therefore $\chi\leq 4(D-1)^2+1$. This coloring yields a decomposition of the edges
\begin{equation}
    \mathcal{E} = \bigsqcup_{c=1}^{\chi}\mathcal{E}_c,
\end{equation}
where $\mathcal{E}_c$ is the set of edges with color $c$.

For each color $c=1,2,\cdots,\chi$, we then learn all the coefficients associated with this color. We denote by $\mathcal{V}_c$ all the vertices (bosonic modes) that are contained in an edge in $\mathcal{E}_c$. During time evolution, we apply random unitaries of the form 
\begin{equation}
    U = \prod_{i\in\mathcal{V}\setminus \mathcal{V}_c} e^{-i\theta_i b_i^{\dagger}b_i},\quad \theta_i\sim\mathcal{U}([0,2\pi]).
\end{equation}
Here $\theta_i$, $i\in \mathcal{V}\setminus \mathcal{V}_c$, are independent random variables. Following the derivation in \eqref{eq:first_order_approx_short_time_step_appendix}, we can see that the effective Hamiltonian is
\begin{equation}
    H_{\mathrm{effective}}=\prod_{i\in\mathcal{V}\setminus \mathcal{V}_c}\left(\frac{1}{2\pi}\int_0^{2\pi}\dd \theta_i \right) e^{-i\sum_{i\in\mathcal{V}\setminus \mathcal{V}_c}\theta_i n_i}H e^{i\sum_{i\in\mathcal{V}\setminus \mathcal{V}_c}\theta_i n_i}.
\end{equation}

We can then examine the effect of this transformation on each term. For a term $b_k^{\dagger}b_l$, $k\neq l$, if $k$ is in $\mathcal{V}\setminus \mathcal{V}_c$ but $l$ is not, then
\begin{equation}
    \prod_{i\in\mathcal{V}\setminus \mathcal{V}_c}\left(\frac{1}{2\pi}\int_0^{2\pi}\dd \theta_i \right) e^{-i\sum_{i\in\mathcal{V}\setminus \mathcal{V}_c}\theta_i n_i} b_k^{\dagger}b_l e^{i\sum_{i\in\mathcal{V}\setminus \mathcal{V}_c}\theta_i n_i} = \frac{1}{2\pi}\int_0^{2\pi} \dd \omega_k e^{i\omega_k} b_k^{\dagger}b_l = 0
\end{equation}
The same is true if $l$ is in $\mathcal{V}\setminus \mathcal{V}_c$ but $k$ is not. When both $k,l\in\mathcal{V}\setminus \mathcal{V}_c$, then
\begin{equation}
\begin{aligned}
    &\prod_{i\in\mathcal{V}\setminus \mathcal{V}_c}\left(\frac{1}{2\pi}\int_0^{2\pi}\dd \theta_i \right) e^{-i\sum_{i\in\mathcal{V}\setminus \mathcal{V}_c}\theta_i n_i} b_k^{\dagger}b_l e^{i\sum_{i\in\mathcal{V}\setminus \mathcal{V}_c}\theta_i n_i} \\
    &= \frac{1}{(2\pi)^2}\int_0^{2\pi} \dd \omega_k \int_0^{2\pi} \dd \omega_l e^{i(\omega_k-\omega_l)} b_k^{\dagger}b_l = 0.
\end{aligned}
\end{equation}
In other words, for any coupling term $b^{\dagger}_k b_l$, the above procedure will cancel it out if either $k$ or $l$ is in $\mathcal{V}\setminus \mathcal{V}_c$. All other terms are preserved because they commute with $n_i$ for any $i\in\mathcal{V}\setminus \mathcal{V}_c$. The only possible $b^{\dagger}_k b_l$ terms left are those with $k,l\in\mathcal{V}_c$. 
This also means that $(k,l)\in\mathcal{E}_c$ because of the following argument: first by definition of $\mathcal{V}_c$ there must exists $k'$ and $l'$ such that $(k,k')\in\mathcal{E}_c$ and $(l,l')\in\mathcal{E}_c$. We must have $(k,l)\in\mathcal{E}$, as otherwise, this coupling term would not exist at all. This means that unless $k'=l'$, the two edges $(k,k')$ and $(l,l')$ as vertices in $\mathcal{L}$ are next-neighbors, which is not allowed in our coloring. Therefore $k'=l'$ and we have $(k,l)\in\mathcal{E}_c$.
Consequently, the effective Hamiltonian is
\begin{equation}
\label{eq:effective_hamiltonian_decoupled}
    H_{\mathrm{effective}} = \sum_{(i,j)\in \mathcal{E}_c} h_{ij}b_i^{\dagger}b_j + \sum_i \omega_i b_i^{\dagger}b_i + \frac{\xi_i}{2}\sum_i n_i(n_i-1).
\end{equation}

Next, we will show that the above Hamiltonian is decoupled into clusters of sizes at most $2$. We will do this by showing that any bosonic mode $i$ interacts with at most one other bosonic mode in the above Hamiltonian. This can be proved by contradiction: if $i$ interacts with both $j$ and $k$ in the above Hamiltonian, then $(i,j)\in\mathcal{E}_c$ and $(i,k)\in\mathcal{E}_c$, which makes $(i,j)$ and $(i,k)$ neighbors as vertices in $\mathcal{L}$, and this is forbidden in our coloring.

With the decoupled Hamiltonian in \eqref{eq:effective_hamiltonian_decoupled}, we can then learn the coefficients in each one- or two-mode cluster independently and in parallel using the algorithms described in Sections~\ref{sec:learning_an_anharmonic_oscillator_appendix} and \ref{sec:learning_two_coupled_anharmonic_oscillators_appendix}. Looping over all colors $c\in\{1,2,\cdots,\chi\}$, we will obtain all the coefficients in the Hamiltonian.

\section{Deviation from the effective dynamics}
\label{sec:deviation_from_effective_dynamics_appendix}
In this section, we consider the error introduced by simulating the effective dynamics with the insertion of random unitaries, as mentioned in Section~\ref{sec:single_mode_coeff_appendix}. Suppose $\mathcal{D}$ is a distribution over the set of unitaries, and the initial state of the system is represented by the density matrix $\rho(0)$. The actual final state obtained after the insertion of $r$ random unitaries is
\begin{equation}\label{eq:qdriftres}
    \E_{U_j\sim\mathcal{D}} \left(\prod_{1\le j\le r}^{\leftarrow}U_j^\dagger e^{-i\tau H} U_j\right)\rho(0)\left(\prod_{1\le j\le r}^{\rightarrow}U_j^\dagger e^{i\tau H} U_j\right),
\end{equation}
where each $U_j$ is inserted after time $\tau=\frac{t}{r}$. On the other hand, the desired final state, which facilitates the subsequent steps of the learning process, is 
\begin{equation}
     e^{-it H_{\mathrm{effective}}}\rho(0) e^{it H_{\mathrm{effective}}} ,
\end{equation}
where $H_{\mathrm{effective}}$ is the effective Hamiltonian:
\[
  H_{\mathrm{effective}} = \mathbb{E}_{U\sim \mathcal{D}} U^{\dagger}HU.
\]
In this section, we provide an analysis of the difference between the two dynamics for a certain class of Hamiltonians and thereby complete the analysis of approximation errors investigated in Section~\ref{sec:learning_an_anharmonic_oscillator_appendix}. For the sake of the Hamiltonians studied in this paper, we consider the Hamiltonians of the following form:
\begin{equation}\label{eq:generalH}
H = \sum_{\braket{i,j}} h_{ij} b_i^{\dagger}b_j + \sum_i\omega_i n_i + \frac{1}{2}\sum_{\braket{jklm}} \xi_{jklm}b_j^{\dagger}b_k^{\dagger}b_lb_m,
\end{equation}
where in the last term we denote by $\braket{jklm}$ the index quadruples such that $\{j,k,l,m\}$ form a connected subgraph in the underlying graph $\mathcal{G}=(\mathcal{V},\mathcal{E})$ of bosonic modes. We begin with a lemma describing the action of these Hamiltonians on the product of coherent states.
\begin{lem}\label{lem:norm Hphi}
Let
\begin{equation}
H = \sum_{\braket{i,j}} h_{ij} b_i^{\dagger}b_j + \sum_i\omega_i n_i + \frac{1}{2}\sum_{\braket{jklm}} \xi_{jklm}b_j^{\dagger}b_k^{\dagger}b_lb_m,
\end{equation}
and 
\begin{equation}\label{eq:statephi}
	\phit = \bigotimes_{i\in\mathcal{V}} \left(e^{-|\alpha_i|^2/2}\sum_{k=0}^{\infty}\frac{\alpha_i^k e^{-i\zeta_{i,k}}}{\sqrt{k!}}\ket{k}_i\right),
\end{equation}
where $\alpha_i$ is a complex number of magnitude $O(1)$, and $\zeta_{i,k}\in\RR$ can be any real number. Then 
\begin{equation}
	\|H\phit\| = O(N\max\{|\xi_{jklm}|, |\omega_i|,|h_{i,j}|\}),\label{eq:hphi_norm}
\end{equation}
and
\begin{equation}
	\|H^2\phit\| = O(N^2(\max\{|\xi_{jklm}|, |\omega_i|,|h_{i,j}|\})^2),\label{eq:h2phi_norm}
\end{equation}
where $N=|\mathcal{V}|+|\mathcal{E}|$. 
\end{lem}
\begin{proof}
    It suffices to prove the result for $H/\max\{|\xi_{jklm}|, |\omega_i|,|h_{i,j}|\}$. Therefore we assume $\max\{|\xi_{jklm}|, |\omega_i|,|h_{i,j}|\}=1$ without loss of generality. Notice that $H$ is the sum of $O(N)$ terms, and each term takes the form $\bd_p b_q$ or $\bd_p b_q\bd_r b_s$, where $p$, $q$, $r$, $s$ may be repeated. We will prove that each term acting on $\ket{\phit}$ yields a state whose norm is $O(1)$. We first demonstrate this for $\norm{\bd_p b_q\bd_r b_s\phit}$. Simple calculation shows
    \begin{equation}
		\begin{aligned}
			&\bd_p b_q\bd_r b_s\phit \\
			= &\bigotimes_{i\notin\{p,q,r,s\}} \left(e^{-|\alpha_i|^2/2}\sum_{k=0}^{\infty}\frac{\alpha_i^k e^{-i\zeta_{i,k}}}{\sqrt{k!}}\ket{k}_i\right)\otimes \bigotimes_{j\in\{p,q,r,s\}} \left(e^{-|\alpha_j|^2/2}\sum_{k=0}^{\infty}\frac{\alpha_j^k e^{-i\zeta_{j,k}}}{\sqrt{k!}}\sqrt{P_j(k)}\ket{k+\sigma_j}_j\right),
		\end{aligned}		
	\end{equation}
	where $P_j$'s are polynomials with $\sum_{j\in\{p,q,r,s\}}\deg P_j = 4$, and $\sigma_j$ is an integer determined by the numbers of $\bd_j$ and $b_j$ in $\bd_p b_q\bd_r b_s$. For example, if $p=q=r=1$, $s=2$, then $P_1(k)=(k+1)^3$, $\sigma_1=1$, $P_2(k)=k$, $\sigma_2=-1$. Straight calculations can show that 
	\begin{equation}
		\begin{aligned}
			&\left\|e^{-|\alpha_j|^2/2}\sum_{k=0}^{\infty}\frac{\alpha_j^k e^{-i\zeta_{j,k}}}{\sqrt{k!}}\sqrt{P_j(k)}\ket{k+\sigma_j}_j\right\|^2\\
			={}& e^{-|\alpha_j|^2}\sum_{k=0}^{\infty}\frac{|\alpha_j|^{2k}}{k!}P_j(k) = Q_j(|\alpha_j|^2) = O(1).
		\end{aligned}		
	\end{equation}
    where $Q_j$ is a polynomial that can be determined by $P_j$, but we do not care about its explicit form. Therefore we have shown that 
    \begin{equation}
        \norm{\bd_p b_q\bd_r b_s\phit}=\sqrt{\prod_{j\in\{p,q,r,s\}}Q_j(|\alpha_j|^2)}=O(1).
    \end{equation}
    Similarly, we can show that $\|\bd_p b_q\phit\|=O(1)$. Therefore \eqref{eq:hphi_norm} is established.

    Next, we will prove \eqref{eq:h2phi_norm}. We can fully expand $H^2$ into $O(N^2)$ terms, each of which has the form $\bd_p b_q \bd_{p'} b_{q'}$, $\bd_p b_q\bd_r b_s\bd_{p'} b_{q'}$, $\bd_p b_q\bd_{p'} b_{q'}\bd_{r'} b_{s'}$, or $\bd_p b_q\bd_r b_s\bd_{p'} b_{q'}\bd_{r'} b_{s'}$. Again, we may go through a similar process as above and conclude that each term acting on $\phit$ yields a state of magnitude $O(1)$.
\end{proof}

Assume that $\ket{\phi_0} = \bigotimes_i \ket{\alpha_i}$ is a product of coherent states, and $\ket{\phi_t}$ is the state obtained by evolving under the effective dynamics $\He$ for time $t$, i.e., $\ket{\phi_t} = e^{-it\He}\ket{\phi_0}$, then $\ket{\phi_t}$ is a state of the form described in \eqref{eq:statephi} for the distribution $\mathcal{D}$ used in previous sections. 
Using density matrices, the effective dynamics with the Hamiltonian $\He$ starts from the state $\rho(0):=\ket{\phi_0}\bra{\phi_0}$ and end up in the state $\rho(t):=\ket{\phi_t}\bra{\phi_t}$ at time $t$, while the actual final state obtained is given by \eqref{eq:qdriftres}.
To bound its distance from the desired state $\rho(t)$, we define the following density operators:
\begin{equation}\label{eq:telescope}
    \rho^{(\ell)}(t) = \E \left(\prod_{1\le j\le \ell}^{\leftarrow}U_j^\dagger e^{-i\tau H} U_j\right)\rho(t-\ell\tau)\left(\prod_{1\le j\le \ell}^{\rightarrow}U_j^\dagger e^{i\tau H} U_j\right).
\end{equation}
Then $\rho^{(0)}(t) = \rho(t)$ and $\rho^{(r)}(t)$ is the density operator in \eqref{eq:qdriftres}. 
Now consider the distance between $\rho^{(L-1)}(t)$ and $\rho^{(L)}(t)$. Define 
\begin{equation}
    Q^{(L)} = \prod_{1\le j\le L-1}^{\rightarrow}U_j^\dagger e^{i\tau H} U_j, 
\end{equation}
then by the independence of $U_j$, we have
\begin{equation}\label{eq:localerror}
\begin{aligned}
&\|\rho^{(L)}(t)-\rho^{(L-1)}(t)\|_* \\
&=\left\|\E_{Q^{(L)}}\left[ Q^{(L)}\left(\E_U\left(U^\dagger e^{-i\tau H} U\rho(t-L\tau)U^\dagger e^{i\tau H} U-e^{-i\tau\He}\rho(t-L\tau)e^{-i\tau\He}\right)\right)(Q^{(L)})^\dagger\right]\right\|_*\\
&\le\E_{Q^{(L)}}\left\|Q^{(L)}\left(\E_U\left(U^\dagger e^{-i\tau H} U\rho(t-L\tau)U^\dagger e^{i\tau H} U-e^{-i\tau\He}\rho(t-L\tau)e^{-i\tau\He}\right)\right)(Q^{(L)})^\dagger\right\|_*\\
&= \E_{Q^{(L)}}\left\|\E_U\left(U^\dagger e^{-i\tau H} U\rho(t-L\tau)U^\dagger e^{i\tau H} U-e^{-i\tau\He}\rho(t-L\tau)e^{-i\tau\He}\right)\right\|_*\\
&=\left\|\E_U\left(U^\dagger e^{-i\tau H} U\rho(t-L\tau)U^\dagger e^{i\tau H} U-e^{-i\tau\He}\rho(t-L\tau)e^{-i\tau\He}\right)\right\|_*,
\end{aligned}
\end{equation}
where$\|\cdot\|_*$ denotes the trace norm (nuclear norm). The fourth line follows from the property of trace norm and the fact that $Q^{(L)}$ is unitary. From the Taylor expansion, one can obtain 
\begin{equation}\label{eq:qdrfttaylor}
\begin{aligned}
&\E\left(U^\dagger e^{-i\tau H} U\rho(t-L\tau)U^\dagger e^{i\tau H} U\right)-\rho(t-L\tau) \\
&= \E\left(e^{-i\tau U^\dagger HU}\rho(t-L\tau) e^{i\tau U^\dagger HU} \right)-\rho(t-L\tau)\\
&= \E\left(-i\tau [U^\dagger H U, \rho(t-L\tau)]-\int_0^\tau e^{-is U^\dagger HU}[U^\dagger HU, [U^\dagger HU, \rho(t-L\tau)]] e^{is U^\dagger HU}(\tau-s) \right)\\
&=-i\tau [\E (U^\dagger H U), \rho(t-L\tau)]-\E\left(\int_0^\tau e^{-is U^\dagger HU}[U^\dagger HU, [U^\dagger HU, \rho(t-L\tau)]] e^{is U^\dagger HU}(\tau-s) \right)\\
&= -i\tau [\He, \rho(t-L\tau)]-\E\left(\int_0^\tau e^{-is U^\dagger HU}[U^\dagger HU, [U^\dagger HU, \rho(t-L\tau)]] e^{is U^\dagger HU}(\tau-s) \right).
\end{aligned}
\end{equation}
Similarly, one has
\begin{equation}\label{eq:qefftaylor}
\begin{aligned}
&\E\left( e^{-i\tau \He} \rho(t-L\tau) e^{i\tau \He} \right)-\rho(t-L\tau) \\
&= -i\tau [\He, \rho(t-L\tau)]-\int_0^\tau e^{-is \He}[\He, [\He, \rho(t-L\tau)]] e^{is \He}(\tau-s) .
\end{aligned}
\end{equation}
Combining \eqref{eq:qdrfttaylor} and \eqref{eq:qefftaylor}, one obtains
\begin{equation}\label{eq:combinetaylor}
\begin{aligned}
&\left\|\E\left(U^\dagger e^{-i\tau H} U\rho(t-L\tau)U^\dagger e^{i\tau H} U-e^{-i\tau\He}\rho(t-L\tau)e^{-i\tau\He}\right)\right\|_*\\
&\le \left\|\E\left(\int_0^\tau e^{-is U^\dagger HU}[U^\dagger HU, [U^\dagger HU, \rho(t-L\tau)]] e^{is U^\dagger HU}(\tau-s) \right)\right\|_*\\
&+\left\|\int_0^\tau e^{-is \He}[\He, [\He, \rho(t-L\tau)]] e^{is \He}(\tau-s)\right\|_*\\
&\le\half\tau^2\bigg(\sup_U \left\|[U^\dagger HU, [U^\dagger HU, \rho(t-L\tau)]] \right\|_*+\left\|[\He, [\He, \rho(t-L\tau)]]\right\|_*\bigg)
\end{aligned}
\end{equation}
One only needs to bound $\left\|[U^\dagger HU, [U^\dagger HU, \rho(t-L\tau)]] \right\|_*$ and $\left\|[\He, [\He, \rho(t-L\tau)]]\right\|_*$. By a direct calculation, one sees that
\begin{equation}\label{eq:Heres}
\begin{aligned}
&\left\|[\He, [\He, \rho(t-L\tau)]]\right\|_*\\
&\le \|\He^2\rho(t-L\tau)\|_* + 2\|\He\rho(t-L\tau)\He\|_*+\|\rho(t-L\tau)\He^2\|_*\\
&=2\|\phi_{t-L\tau}\|\|\He^2\phi_{t-L\tau}\|+2\|\He\phi_{t-L\tau}\|^2\\
&\le C N^2 \max\{|\xi_{jklm}|, |\omega_i|, |h_{i,j}|\}^2,\\
\end{aligned}
\end{equation}
where $C=\mathcal{O}(1)$ is a constant, and we have used the property of the trace norm for rank-$1$ matrices. In the last step, we are using \Cref{lem:norm Hphi} with $H=\He$ and $\phit=\phi_{t-L\tau}$. Similarly, one can obtain
\begin{equation}\label{eq:Hres}
\begin{aligned}
&\left\|[U^\dagger HU, [U^\dagger HU, \rho(t-L\tau)]]\right\|_*\\
&=2\|\phi_{t-L\tau}\|\|U^\dagger H^2U\phi_{t-L\tau}\|+2\|U^\dagger HU\phi_{t-L\tau}\|^2\\
&=2\|H^2U\phi_{t-L\tau}\|+2\|HU\phi_{t-L\tau}\|^2\\
&\le C N^2 \max\{|\xi_{jklm}|, |\omega_i|, |h_{i,j}|\}^2.\\ 
\end{aligned}
\end{equation}
In the last step, we are using \Cref{lem:norm Hphi} with $H=H$ and $\phit=U\phi_{t-L\tau}$. As a result, we have proved the following:

\begin{thm}
    For a Hamiltonian of the form described in \eqref{eq:generalH} and a product of coherent states $\ket{\phi_0} = \otimes_i \ket{\alpha_i}$ such that $\alpha_i$ are $\mathcal{O}(1)$ constants, we have 
    \begin{equation}\label{eq:qdriferror}
    \begin{aligned}
    &\Bigg\|\E_{U_j\sim\mathcal{D}} \left(\prod_{1\le j\le r}^{\leftarrow}U_j^\dagger e^{-i\tau H} U_j\right)\rho(0)\left(\prod_{1\le j\le r}^{\rightarrow}U_j^\dagger e^{i\tau H} U_j\right) - e^{-it H_{\mathrm{effective}}}\rho(0) e^{it H_{\mathrm{effective}}}\Bigg\|_*\\
    &\le C N^2 \frac{t^2}{r} \max\{|\xi_{jklm}|, |\omega_i|, |h_{i,j}|\}^2,
    \end{aligned}
\end{equation}
where $\rho(0) = \ket{\phi_0}\bra{\phi_0}$, $H_{\mathrm{effective}} = \mathbb{E}_{U\sim \mathcal{D}} U^{\dagger}HU$, $C$ is a $\mathcal{O}(1)$ constant, $N=|\mathcal{V}|+|\mathcal{E}|$ and $\mathcal{G}=(\mathcal{V},\mathcal{E})$ is the underlying graph of bosonic modes.
\end{thm}
\begin{proof}
The left-hand side of \eqref{eq:qdriferror} can be expressed by $\|\rho^{(r)}(t)-\rho^{(0)}(t)\|_*$, where $\rho^{(r)}(t)$ and $\rho^{(0)}(t)$ are defined in \eqref{eq:telescope}. Thus
\begin{equation}
\begin{aligned}
&\|\rho^{(r)}(t)-\rho^{(0)}(t)\|_*\le\sum_{L=1}^r\|\rho^{(L)}(t)-\rho^{(L-1)}(t)\|_*\\
&\le \sum_{L=1}^r C N^2 \tau^2 \max\{|\xi_{jklm}|, |\omega_i|, |h_{i,j}|\}^2 = C N^2 \frac{t^2}{r} \max\{|\xi_{jklm}|, |\omega_i|, |h_{i,j}|\}^2,
\end{aligned}
\end{equation}
where we have used \eqref{eq:combinetaylor}, \eqref{eq:Heres} and \eqref{eq:Hres} in the second inequality.
\end{proof}

\end{document}